\documentclass[12pt,a4paper]{article}

\usepackage{amsmath,amsthm,amssymb,amscd,a4wide}

\usepackage{dsfont}
\usepackage{mathrsfs}
\usepackage{setspace}

\newcommand{\ud}{\mathrm{d}}

\newcommand{\ue}{\mathrm{e}}

\newcommand{\cD}{{\mathcal D}}
\newcommand{\cE}{{\mathcal E}}

\newcommand{\cH}{{\mathcal H}}
\newcommand{\cP}{{\mathcal P}} 
\newcommand{\cV}{{\mathcal V}}
\newcommand{\M}{\mathrm{M}}
\newcommand{\GL}{\mathrm{GL}}

\newcommand{\rz}{{\mathbb R}}
\newcommand{\nz}{{\mathbb N}}
\newcommand{\kz}{{\mathbb C}}

\DeclareMathOperator{\ran}{ran}
\DeclareMathOperator{\rank}{rk}

\newcommand{\eins}{\mathds{1}}

\numberwithin{equation}{section}

\newtheorem{theorem}{Theorem}[section]
\newtheorem{lemma}[theorem]{Lemma}
\newtheorem{prop}[theorem]{Proposition}

\theoremstyle{definition}
\newtheorem{defn}[theorem]{Definition}
\newtheorem{rem}[theorem]{Remark}

\numberwithin{equation}{section}

\begin{document}

\thispagestyle{empty}

\vspace*{0.5cm}

\begin{center}

{\LARGE\bf Quantum graphs with singular\\[5mm]
two-particle interactions} \\

\vspace*{3cm}

{\large Jens Bolte}%
\footnote{E-mail address: {\tt jens.bolte@rhul.ac.uk}}
{\large and Joachim Kerner}%
\footnote{E-mail address: {\tt joachim.kerner.2010@live.rhul.ac.uk}}
\vspace*{2cm}

Department of Mathematics\\
Royal Holloway, University of London\\
Egham, TW20 0EX\\
United Kingdom\\

\end{center}

\vfill

\begin{abstract}
We construct quantum models of two particles on a compact metric graph with 
singular two-particle interactions. The Hamiltonians are self-adjoint 
realisations of Laplacians acting on functions defined on pairs of edges in 
such a way that the interaction is provided by boundary conditions. In order
to find such Hamiltonians closed and semi-bounded quadratic forms are
constructed, from which the associated self-adjoint operators are extracted.
We provide a general characterisation of such operators and, furthermore,
produce certain classes of examples. We then consider identical particles
and project to the bosonic and fermionic subspaces. Finally, we show that the 
operators possess purely discrete spectra and that the eigenvalues are 
distributed following an appropriate Weyl asymptotic law.   
\end{abstract}

\newpage

\section{Introduction}
Initially, quantum graphs were introduced by Ruedenberg and Scherr 
\cite{RueSch53} as simplified models to describe the electronic structure of 
aromatic hydrocarbons, where the graph connectivity follows from the 
bond structure of the atoms in the molecule. Since then numerous variants of 
quantum graphs have been developed and have found applications in many areas 
of physics and mathematics (see \cite{Exnetal08} for a review). 

Quantum systems on graphs have proven to be attractive and versatile models 
as they combine the simplicity of one (or zero) spatial dimension with the 
complexity induced by the connectivity of the underlying graph. In the 
following we shall focus on models in which one-dimensional quantum systems 
associated with the edges of a (metric) graph are coupled in the vertices 
of the graph. Alternative models utilise finite-dimensional quantum systems 
associated with the (zero-dimensional) vertices of a graph. 

In mathematical terms, the stationary Schr\"odinger equation of a quantum graph
is a system of coupled, ordinary differential equations, describing the quantum
motion on the edges of the graph that are connected in the vertices. Quite 
surprisingly, such coupled systems of equations can be used to interpolate 
between single, ordinary differential equations describing pure, 
one-dimensional quantum systems on the one hand, and partial differential 
equations as they arise in multi-dimensional quantum systems on the other hand.
Indeed, Kottos and Smilansky observed that the eigenvalue correlations in 
quantum graphs generically are the same as in quantum systems with chaotic 
classical limit \cite{KotSmi99}. The latter require at least two spatial 
dimensions and thus are modelled by partial differential equations. Following 
the random matrix conjecture of quantum chaos \cite{BohGiaSch84}, these 
correlations can be described by eigenvalue correlations of suitable random 
matrices. Many further studies along these lines have followed since (see 
\cite{GnuSmi07} for a review).

So far quantum graphs mostly have been used to model one-particle systems.
Harmer, however, introduced $\delta$-type two-particle interactions on the
edges \cite{Har07a,Har08}, and Harrison et al.\ studied the particle exchange
symmetry in many-particle versions of finite-dimensional quantum graph
models \cite{HarKeaRob11}. In this paper we shall study two-particle 
quantum systems on graphs. Our focus will be the introduction of genuine 
two-particle interactions along the lines of the one-particle interactions 
(with the outside world) present in the existing quantum graph models. The 
latter represent singular interactions (`potentials') that are strictly 
localised in the vertices of the graph so that the motion along edges is free. 
The associated quantum Hamiltonians are constructed as self-adjoint extensions 
of a symmetric realisation of the Laplacian acting on functions on the edges. 
These extensions introduce coupling conditions in the vertices and hence 
provide localised interactions. Following the tensor-product construction of 
many-particle quantum models, two-particle states on graphs are composed of 
functions on pairs of edges, one for each particle, and the free Hamiltonian 
is a Laplacian acting on functions defined on rectangles. The natural, 
symmetric realisation of such a Laplacian, however, has infinite deficiency 
indices so that a construction of self-adjoint extensions is less straight 
forward than in the one-particle case. Therefore, we first define suitable, 
closed and semi-bounded quadratic forms and then extract from such a form the 
unique, semi-bounded and self-adjoint operator that is associated with it. 
This operator is, of course, again a Laplacian acting on functions on 
rectangles, but with boundary conditions imposed along the edges of the 
rectangles. In a general setting these boundary conditions emerge in a weak 
form, but under specific circumstances related to elliptic regularity we are 
able to provide explicit versions of these conditions that are closely related 
to the respective conditions in a one-particle quantum graph. The resulting
two-particle interactions are singular and, in contrast to Harmer's
construction \cite{Har07a,Har08}, localised in vertices rather than on edges.

Our main results are, first, the construction of self-adjoint operators that 
describe singular two-particle interactions on graphs and, second, proving 
that these operators possess purely discrete spectra and that their eigenvalue 
count follows a suitable Weyl asymptotic law. 

This paper is organised as follows: In Section~\ref{sec1} we briefly recall 
the construction of one-particle quantum graphs as well as some of their 
properties that are relevant for our purposes. Then we describe the 
tensor-product construction of systems of two distinguishable as well as of 
two identical particles as applied to quantum graphs. In Section~\ref{sec2} 
we first perform the construction of self-adjoint operators for two 
distinguishable particles on an interval via a suitable quadratic form. We 
then extend these constructions to arbitrary compact metric graphs.
We also characterise those among the self-adjoint operators that represent 
non-trivial two-particle interactions. The necessary adaptions to systems
of two identical particles (bosons or fermions) on general compact graphs
are performed in Section~\ref{sec3}. In Section~\ref{sec4} we prove that
the two-particle Hamiltonians constructed before have compact resolvent and,
therefore, possess purely discrete spectra; their eigenvalue count follows a 
Weyl-type asymptotic law. We defer the proof our main regularity result in 
Section~\ref{sec2} to an appendix.
%
%
%
%
\section{Preliminaries}
\label{sec1}
Quantum systems with many particles are obtained from tensor
product constructions based on the underlying one-particle systems.
For that reason we briefly summarise one-particle quantum graphs,
and then explain the basic steps to construct quantum systems of two
(identical) particles on a graph. For more details on one-particle
quantum graphs see
\cite{KotSmi99,KosSch99,Kuc04,GnuSmi07,Exnetal08,BolEnd09}.
\subsection{One-particle quantum graphs}
The classical configuration space of a quantum graph is a compact
metric graph, i.e., a finite graph $\Gamma = (\cV,\cE)$ with vertices
$\cV = \{v_1,\dots,v_V\}$ and edges $\{e_1,\dots,e_E\}$. The latter
are identified with intervals $[0,l_e]$, $e=1,\dots,E$, thus introducing
a metric on the graph. At this point we do not exclude multiple edges
and loops.

Functions on the graph are collections of functions on the edges,
i.e.,
\begin{equation}
 F=(f_1,\dots,f_E) \ ,\quad\text{with}\quad f_e : [0,l_e]\to\kz\ ,
\end{equation}
so that spaces of functions on $\Gamma$ are (finite) direct sums of
the respective spaces of functions on the edges. The most relevant
space is the one-particle Hilbert space
\begin{equation}
 \cH_1 = L^2 (\Gamma) := \bigoplus_{e=1}^E L^2 (0,l_e) \ ,
\end{equation}
and all other spaces are constructed in a similar way.

One-particle observables are self-adjoint operators on $\cH_1$,
among which the Hamiltonian plays a prominent role. In the absence
of external forces or gauge fields the Hamiltonian should be a suitable
version of a Laplacian. As a differential operator the (positive) Laplacian
acts according to
\begin{equation}
 -\Delta_1 F= (-f_1'',\dots,-f_E'')
\end{equation}
on $F\in C^\infty (\Gamma)$. We here use the index to indicate that this
is a one-particle Laplacian.

Viewed as an operator in $L^2 (\Gamma)$ with domain
$C^\infty_0 (\Gamma)$, this Laplacian is symmetric,
but not self-adjoint. One can construct and classify all self-adjoint
extensions of this operator using von Neumann's theory. In the context
of quantum graphs, however, an alternative parametrisation of self-adjoint
extensions  in terms of linear relations among the boundary values
\begin{equation}
\label{Fbv}
 F_{bv} := \bigl( f_1(0),\dots,f_E (0),f_1(l_1),\dots,f_E (l_E) \bigr)^T
              \in\kz^{2E}\ ,
\end{equation}
of functions and (inward) derivatives,
\begin{equation}\label{Fbv2}
 F_{bv}' := \bigl( f_1'(0),\dots,f_E' (0),-f_1'(l_1),\dots,-f_E' (l_E) \bigr)^T
               \in\kz^{2E}\ ,
\end{equation}
has proven useful. Kostrykin and Schrader \cite{KosSch99} proved
the following.
\begin{theorem}[Kostrykin, Schrader]
\label{KosSchThm}
Any self-adjoint realisation of the Laplacian on a compact, metric graph
has a domain of the form
\begin{equation}
\label{1partBC}
 \cD_1 (A,B) = \{ F\in H^2 (\Gamma);\  AF_{bv} +BF_{bv}' =0 \}  \ ,
\end{equation}
where $A,B\in\M(2E,\kz)$ are such that $\rank (A,B)=2E$ and
$A B^\ast$ is self-adjoint.

Moreover, two such realisations, with domains $\cD(A,B)$ and
$\cD(A',B')$, are equivalent, iff there exists $C\in\GL(2E,\kz)$ such
that $A'=CA$ and $B'=CB$.
\end{theorem}
An alternative characterisation of the domain \eqref{1partBC} employs the
orthogonal projectors $P$ onto $\ker B\subset\kz^{2E}$ and $Q=\eins_{2E} -P$,
as well as the self-adjoint endomorphism $L=(B|_{\ran B^\ast})^{-1}AQ$ of 
$\ran Q\subset\kz^{2E}$. Kuchment showed \cite{Kuc04} that the domain 
$\cD_1 (A,B)$ is the same as
\begin{equation}
\label{1partBCalt}
 \cD_1 (P,L) = \{ F\in H^2 (\Gamma);\  PF_{bv}=0\ \text{and}\
                  QF_{bv}'+LQF_{bv}=0 \}  \ .
\end{equation}
This way self-adjoint realisations of the Laplacian are uniquely characterised
in terms of the projector $P$ and the self-adjoint map $L$. From now on
we shall adhere to this parametrisation of domains.

Yet another way of specifying self-adjoint Laplacians is in terms of their
associated quadratic forms \cite{Kuc04}.
\begin{theorem}[Kuchment]
\label{KuchmentThm}
The quadratic form associated with a Laplacian $-\Delta_1$ defined on
the domain $\mathcal{D}_1(P,L)$ is
\begin{equation}
\label{Qform1}
\begin{split}
 Q^{(1)}_{P,L}[F]
  &= \int_\Gamma |\nabla f| \ \ud x -
        \langle F_{bv},LF_{bv}\rangle_{\kz^{2E}}   \\
  &= \sum_{e=1}^E \int_0^{l_e} |f'_e (x)|^2 \ dx -
          \langle F_{bv},LF_{bv}\rangle_{\kz^{2E}}\ ,
\end{split}
\end{equation}
with form domain
\begin{equation}
\label{Qformdomain}
 \cD_{Q^{(1)}} = \{F\in H^1(\Gamma);\ PF_{bv}=0\}\ .
\end{equation}
\end{theorem}
In all of the above the boundary conditions imposed on functions in the
domains are such that, in principle, they relate all boundary values among
each other. Whenever the boundary conditions only relate boundary values at
edge ends that are connected in a given vertex, the boundary conditions are
said to be {\it local}. In that case the linear maps $A,B,P,Q,L$ are
block-diagonal with respect to the decomposition
\begin{equation}
\label{1localbc}
 \kz^{2E} = \bigoplus_{v\in\cV}\kz^{d_v}
\end{equation}
of the space of boundary values, where $d_v$ is the degree of the vertex $v$.

When boundary conditions are local one can view them as representing
local, singular interactions of the particle on the graph with an `external
potential' in the vertices. Non-local boundary conditions would model non-local,
singular forces acting on the particle and are, therefore, often discarded.
\subsection{Two-particle quantum graphs}
Placing two particles on a graph first requires to introduce a
two-particle Hilbert space. For two distinguishable particles this is the
tensor product of two one-particle Hilbert spaces,
\begin{equation}
 \cH_2 := \cH_1 \otimes \cH_1 \ .
\end{equation}
For a quantum graph this means that
\begin{equation}
 \cH_2 := \Bigl(\bigoplus_{e=1}^E L^2 (0,l_e)\Bigr) \otimes
 \Bigl(\bigoplus_{e=1}^E L^2 (0,l_e)\Bigr) \ ,
\end{equation}
such that vectors $\Psi\in\cH_2$ are collections $\Psi = (\psi_{e_1 e_2})$
of $E^2$ functions
\begin{equation}
 \psi_{e_1 e_2} \in L^2 (0,l_{e_1}) \otimes L^2 (0,l_{e_2}) \ .
\end{equation}
Two-particle observables are self-adjoint operators in $\cH_2$.
By the above, these are given in components as,
\begin{equation}
 (O\Psi)_{e_1 e_2} = \sum_{f_1 ,f_2 =1}^E
 O_{e_1 e_2,f_1 f_2}\psi_{f_1 f_2} \ .
\end{equation}
A particular set of two-particle observables are those that are
given as lifts of one-particle observables. If, for simplicity we restrict
our attention to bounded operators, a one-particle observable $O_1$
can be lifted to a two-particle observable as
\begin{equation}
\label{factorisedobs}
 O_2 := O_1 \otimes\eins_{\cH_1} + \eins_{\cH_1} \otimes O_1\ .
\end{equation}
Unbounded operators allow for an equivalent construction (see \cite{ReeSim72}).
Any observable of this kind does not represent interactions, or correlations,
between the particles.

On a formal level, the one-particle Laplacian has an obvious
lift to a two-particle operator $-\Delta_2$; its operator-matrix entries
read
\begin{equation}
\label{NLaplace}
 -\Delta_{2,e_1 e_2} = -\frac{\partial^2}{\partial x_{e_1}^2}
 -\frac{\partial^2}{\partial x_{e_2}^2}\ ,
\end{equation}
and hence have the same form as a Laplacian in $\rz^2$. Defined on
the domain $C^\infty_0 (\Gamma)\otimes C^\infty_0 (\Gamma)$,
this operator is symmetric, but not self-adjoint. Again, as in the one-particle
case, the self-adjoint extensions of this operator are observables, and
hence candidates for a two-particle Hamiltonian.

Below we shall see that among the self-adjoint realisations of the
two-particle Laplacian we can identify classes of operators that, indeed,
are lifts of one-particle Laplacians, and others that are not. The latter
represent genuine two-particle interactions and their identification and
characterisation is the principal goal of this paper.

We also want to consider identical particles. This means that a particle 
exchange is a symmetry of the quantum system and hence the symmetric group 
$S_2$ has to be represented unitarily on the two-particle Hilbert space. 
Following the symmetrisation postulate for a system of $N$ identical particles,
the two physically relevant irreducible representations of $S_N$ are the 
totally symmetric and the totally anti-symmetric representation; according to 
the spin-statistic theorem, these representations correspond to bosons and 
fermions, respectively. When $N=2$ these are the only unitary irreducible 
representations anyway. The bosonic, i.e., the totally symmetric representation
is defined on the bosonic two-particle Hilbert space $\cH_{2,B}$, which is the 
symmetric tensor product of two one-particle spaces. Hence, 
$\Psi =(\psi_{e_1 e_2})\in\cH_{2,B}$, iff
\begin{equation}
 \psi_{e_1 e_2}(x_{e_1},x_{e_2}) =
 \psi_{e_2 e_1}(x_{e_2},x_{e_{1}})\ .
\end{equation}
The projection $\Pi_s :\cH_2\to\cH_{2,B}$ then reads
\begin{equation*}
 (\Pi_s\Psi)_{e_1 e_2} = \frac{1}{2}\bigl( \psi_{e_1 e_2} + \psi_{e_2 e_1}
 \bigr)\ .
\end{equation*}
Similarly, the fermionic, totally antisymmetric representation is defined
on the fermionic two-particle Hilbert space $\cH_{2,F}$, which is the 
anti-symmetric tensor product of $\cH_1$ with itself. Accordingly,
$\Psi =(\psi_{e_1 e_2})\in\cH_{2,F}$, iff
\begin{equation}
 \psi_{e_1 e_2}(x_{e_1},x_{e_2}) =
 - \psi_{e_2 e_1}(x_{e_2},x_{e_{1}})\ .
\end{equation}
The projection $\Pi_a :\cH_2\to\cH_{2,F}$ then reads
\begin{equation*}
 (\Pi_a\Psi)_{e_1 e_2} = \frac{1}{2}\bigl( \psi_{e_1 e_2} - \psi_{e_2 e_1}
 \bigr)\ .
\end{equation*}
%
%
%
%
\section{Two-particle interactions}
\label{sec2}
We now introduce singular two-particle interactions. In a first step
we perform a detailed analysis of two distinguishable, interacting particles
on an interval. This is then generalised to two distinguishable particles on 
a general compact, metric graph.
%
\subsection{Two distinguishable particles on an interval}
\label{2partint}
As a first step towards our principal goal we now address the most
simple graph consisting of two vertices and one edge, i.e., an
interval $[0,l]$. In that case all two-particle functions are
defined on the square $D:= (0,l)\times (0,l)$. Hence, in particular,
the Hilbert space for two distinguishable particles is
$\cH_2 = L^2(0,l)\otimes L^2(0,l) = L^2 (D)$.

The goal now is to find self-adjoint realisations of the Laplacian
with a domain that is a subspace of $L^2(D)$. In the one-particle case
self-adjoint realisations of the Laplacian were obtained as maximally
symmetric extensions of a suitable symmetric realisation of the Laplacian.
In the two-particle case the corresponding Laplacian $-\Delta_{2,0}$ has a
domain $C_0^\infty (D)$. Hence, the domain of its adjoint $-\Delta_{2,0}^\ast$ is
\begin{equation}
\label{adjdom}
 \cD(-\Delta_{2,0}^\ast) = \{ \psi\in L^2 (D);\ \exists\chi\in L^2 (D)\
 \text{s.t.}\ \langle\psi,-\Delta_{2,0}\phi\rangle = \langle\chi,\phi\rangle
 \ \forall\phi\in C_0^\infty (D)\} \ .
\end{equation}
Notice here that in general $\cD(-\Delta_{2,0}^\ast)\neq H^2(D)$, but
$H^2(D) \subset \cD(-\Delta_{2,0}^\ast)$. Whether this implies that a domain of a
self-adjoint realisation of the Laplacian is in $H^2(D)$ is a subtle issue,
related to the problem of (elliptic) regularity, and will be addressed in
detail below.

Although the situation is similar to the one-particle case, one cannot
proceed to classify self-adjoint extensions of $-\Delta_{2,0}$ in the
same way as the deficiency indices of $-\Delta_{2,0}$ are infinite. It is,
therefore, not automatically guaranteed that the maximally symmetric
extensions of $-\Delta_{2,0}$ are self-adjoint, see \cite{ReeSim79}.

Nevertheless, for the following it will be useful to generate a certain class
of extensions of $-\Delta_{2,0}$. Their domains are subsets of
$\cD(-\Delta_{2,0}^\ast)$ and, in close analogy to the one-particle case, shall
be characterised in terms of boundary conditions imposed on the functions.
These will involve the boundary values (traces) of functions $\psi\in H^1(D)$
as well of their derivatives (in which case $\psi$ has to be in $H^2(D)$)
that, for convenience, are arranged as follows,
\begin{equation}
\label{bvinterval}
 \psi_{bv}(y) = \begin{pmatrix} \psi(0,y) \\ \psi(l,y) \\ \psi(y,0) \\
                 \psi(y,l) \end{pmatrix}   \qquad\text{and}\qquad
 \psi'_{bv}(y) = \begin{pmatrix} \psi_x (0,y) \\ -\psi_x (l,y) \\ \psi_y (y,0)
                 \\ -\psi_y (y,l) \end{pmatrix}\ ,
\end{equation}
i.e., as functions in $L^2(0,l)\otimes\kz^4$. We also require maps
$P,L: [0,l] \to \M(4,\kz)$ fulfilling
\begin{enumerate}
\item $P(y)$ is an orthogonal projector,
\item $L(y)$ is self-adjoint endomorphism of $\ker P(y)$,
\end{enumerate}
for a.e. $y \in [0,l]$. These maps shall be (at least) measurable and bounded. 

With these maps, as well as with $Q(y):=\eins_4-P(y)$, we define the domains
\begin{equation}
\label{bcinterval}
\begin{split}
 \cD_2 (P,L) := \{
       &\psi\in H^2(D);\ P(y)\psi_{bv}(y)=0\ \text{and}\\
       &\quad Q(y)\psi'_{bv}(y)+L(y)Q(y)\psi_{bv}(y)=0\ \text{for a.e.}\
          y\in [0,l] \} \ ,
\end{split}
\end{equation}
that will be useful later on.

In the same way as for one-particle Laplacians, an equivalent characterisation
in terms of maps $A,B: [0,l]\to \M(4,\kz)$ is available, see \eqref{1partBC}
vs.\ \eqref{1partBCalt}. These maps are required to fulfil for a.e.\
$y\in[0,l]$ that $\rank (A(y),B(y))=4$ and that $A(y) B(y)^\ast$ is
self-adjoint. In that case $P(y)$ is a projector onto $\ker B(y)\subseteq\kz^4$
and the self-adjoint map is given by $L(y)=(B(y)|_{\ran B(y)^\ast})^{-1}A(y)Q(y)$ 
on $\kz^4$; indeed, it is an endomorphisms of $\ran B(y)^\ast=\ran Q(y)$.

In a next step we generate a closed and semi-bounded quadratic form that 
allows us to define self-adjoint realisations of the two-particle Laplacian. 
Before, however, we introduce some useful notation: The maps 
$P,L:[0,l]\to\M(4,\kz)$ define (multiplication) operators $\Pi$ and $\Lambda$, 
respectively, on $L^2 (0,l)\otimes\kz^4$ through $(\Pi\chi)(y):=P(y)\chi(y)$ 
and $(\Lambda\chi)(y):=L(y)\chi(y)$, $\chi\in L^2 (0,l)\otimes\kz^4$. As the 
functions $P$ and $L$ are bounded and measurable on $[0,l]$, the operators
$\Pi$ and $\Lambda$ are bounded; $\Pi$ is a projector and $\Lambda$ is 
self-adjoint.

The quadratic form then will derive from the sesqui-linear form
\begin{equation}
\label{Qform2}
\begin{split}
 Q^{(2)}_{P,L}[\psi,\phi]
   &:= \langle  \nabla\psi,\nabla\phi \rangle_{L^2 (D)} - \langle \psi_{bv},
         \Lambda\phi_{bv} \rangle_{L^2(0,l)\otimes\kz^4} \\
   &=  \int_0^l\int_0^l \Bigl(\overline{\psi_x(x,y)}\,\phi_x(x,y) +
          \overline{\psi_y(x,y)}\,\phi_y(x,y)\Bigr)\ \ud x\,\ud y \\
   &\quad-\int_0^l \langle \psi_{bv}(y),L(y)\phi_{bv}(y) \rangle_{\kz^4} \ 
        \ud y \ ,
\end{split}
\end{equation}
as $Q^{(2)}_{P,L}[\psi]:=Q^{(2)}_{P,L}[\psi,\psi]$. For simplicity we use the
same symbol for both forms as it will be clear from the context which form
is meant.
\begin{theorem}
\label{2quadform}
Given maps $P,L: [0,l] \to \M(4,\kz)$ as above that are bounded and measurable,
the quadratic form $Q^{(2)}_{P,L}[\cdot]$ with domain
\begin{equation}
\label{Defquad}
 \cD_{Q^{(2)}} = \{ \psi \in H^1(D);\ P(y)\psi_{bv}(y)=0\ \text{for a.e.}\
 y \in [0,l]\}
\end{equation}
is closed and semi-bounded.
\end{theorem}
\begin{proof}
As $L(y)$ is self-adjoint, the expression \eqref{Qform2} obviously defines
a quadratic form on the domain $\eqref{Defquad}$ in $L^2(D)$. We then observe 
that
\begin{equation}
 \left| \int_0^l \langle \psi_{bv}(y),L(y)\psi_{bv}(y) \rangle_{\kz^{4}}\
 \ud y \right| \leq  L_{max}\, \|\psi_{bv}\|^{2}_{L^{2}(0,l)\otimes\kz^{4}} \ ,
\end{equation}
where
\begin{equation}
 L_{max}:=\sup_{y\in [0,l]}\|L(y)\|_{op} \ .
\end{equation}
Moreover, as a consequence of Lemma 8 in \cite{Kuc04},
\begin{equation}
\label{normeqI}
 \|\psi_{bv}\|^{2}_{L^{2}(0,l)\otimes\kz^{4}} \leq 4\,\left( \frac{2}{\delta}\,
 \|\psi\|^{2}_{L^{2}(D)} + \delta\, \|\nabla\psi\|^{2}_{L^{2}(D)}
 \right) \ ,
\end{equation}
holds for any $\delta\leq l$. Therefore,
\begin{equation}
 Q^{(2)}_{P,L}[\psi] \geq \bigl( 1 -4\delta L_{max}\bigr)\,
 \|\nabla\psi\|^{2}_{L^{2}(D)} - \frac{8L_{max}}{\delta}
 \|\Psi\|^{2}_{L^{2}(D)} \ .
\end{equation}
Now choose $\delta\leq\frac{1}{4 L_{max}}$, then there obviously exits $C>0$
such that 
\begin{equation}
 Q^{(2)}_{P,L}[\psi]\geq -C\|\psi\|^{2}_{L^{2}(D)} 
\end{equation}
and hence the quadratic form is bounded from below. We denote the optimal
such constant by $C_\infty$.

In order to show that the quadratic form \eqref{Qform2} is closed we observe
that the (squared) form norm
\begin{equation}
\label{qformnorm}
\begin{split}
 \|\cdot\|^{2}_{Q^{(2)}_{P,L}}= Q^{(2)}_{P,L}[\cdot]+ (C_{\infty}+1)\,\|\cdot\|^2_{L^2(D)}
\end{split}
\end{equation}
is equivalent to the Sobolev norm in $H^1(D)$. This follows from 
\eqref{normeqI}. Therefore, due to the completeness of $H^1(D)$ any Cauchy 
sequence $\{\psi_{n}\}_{n \in \mathbf{N}}$ in $\cD_{Q^{(2)}}\subset H^1(D)$ with 
respect to the form-norm has a limit $\psi\in H^1(D)$. In order to see that 
this limit is also in $\cD_{Q^{(2)}}$ we recall that according to the trace 
theorem (see \cite{Nec67,Dob05} where Lipschitz domains are covered) there 
exists a constant $c >0$ (depending on the domain $D$), such that
\begin{equation}
\label{trThm}
 \|\gamma\phi\|_{L^{2}(\partial D)} < c \, \|\phi\|_{H^{1}(D)} 
\end{equation}
for all $\phi\in H^1(D)$, where $\gamma: H^1(D)\to L^2(\partial D)$ is the 
trace map, assigning boundary values on $\partial D$ to functions on $D$. In 
our notation, this trace map is effectively given by the expression on the left 
in \eqref{bvinterval}, and we will therefore use $\phi_{bv}$ and $\gamma\phi$ 
interchangeably. (We remark that in our context the estimate \eqref{trThm} 
immediately follows from \eqref{normeqI}.) Thus, $\{\psi_{n,bv}\}$ converges to 
$\psi_{bv}$ in $L^{2}(0,l)\otimes\kz^4$. As the operator 
$\Pi$ on $L^{2}(0,l)\otimes\kz^4$ is 
supposed to be bounded, one concludes that $P(\cdot)\psi_{n;bv}=0$ converges 
to $P(\cdot)\psi_{bv}$ and hence $P(y)\psi_{bv}(y)=0$ for a.e.\ $y\in [0,l]$. 
\end{proof}
We shall now identify the self-adjoint operator $H$ with domain
$\cD(H)\subset\cD_{Q^{(2)}}$ that derives from this quadratic form according to
the representation theorem for quadratic forms (see, e.g., \cite{Kat66}).
In order to specify $H$ and its domain we use that for all $\phi\in\cD(H)$
there exists a unique $\chi\in L^2(D)$, depending on $\phi$, such that for all
$\psi\in\cD_{Q^{(2)}}$ the sesqui-linear form is
\begin{equation}
\label{abstractBVP}
 Q^{(2)}_{P,L}[\phi,\psi] = \langle \chi,\psi\rangle_{L^2(D)} \ .
\end{equation}
We then need to find $H$ and $\cD(H)$ such that $\chi=H\phi$ for all
$\phi\in\cD(H)$.

Our first approach is based on the identification of \eqref{abstractBVP} as
an abstract boundary value problem, following \cite{Sho77}. To this end we
split the sesquilinear form as
\begin{equation}
 Q^{(2)}_{P,L}[\phi,\psi] = q_1 [\phi,\psi] + q_2 [\phi,\psi] \ ,
\end{equation}
where $q_1$ and $q_2$ are, in an obvious way, given by the two terms in
\eqref{Qform2}. The second part is a boundary contribution and, strictly
speaking, involves the linear, continuous trace map
$\gamma: H^1 (D)\to L^{2}(\partial D)$. 

The abstract boundary value problem requires an abstract Green's operator
$\partial_n$ that is constructed as follows (see, e.g., \cite{Sho77}): The
trace map, restricted to the Hilbert space $\cD_{Q^{(2)}}$ (equipped with the
form-norm), has kernel $\ker\gamma = \cD_{Q^{(2)}}\cap H^1_0 (D)$. Let
\begin{equation}
 \cD_0 := \{ \psi\in\cD_{Q^{(2)}};\ \Delta_2\psi\in L^2 (D) \} \ ,
\end{equation}
then $\partial_n:\cD_0\to (\ran\gamma)'$ is a linear map defined by the 
relation
\begin{equation}
\label{abstract1Green}
 q_1[\psi,\phi] - \langle -\Delta_2\psi,\phi\rangle_{L^2(D)} =
 \partial_n\psi [\gamma\phi] \ ,\qquad \phi\in\cD_{Q^{(2)}} \ .
\end{equation}
Notice that whenever $\psi\in H^{2}(D)$ we have
\begin{equation}
\label{normalderiv}
 \partial_n\psi[\gamma\phi] = \int_{\partial D}\frac{\partial \bar{\psi}}
 {\partial n}\,\phi\ \ud s\ ,
\end{equation}
so that the operator $\partial_n$ is the standard normal derivative and 
\eqref{abstract1Green} is the classical first Green's theorem. In the general 
case $\partial_n$ can be seen as a weak form of a normal derivative, thus 
justifying our notation.

We are now in a position to apply Theorem 3.A from \cite{Sho77} to
\eqref{abstractBVP} which yields the following result.
\begin{prop}
\label{abstrdomain}
Let $H$ be the unique self-adjoint, semi-bounded operator corresponding to the
quadratic form $Q^{(2)}_{P,L}$. Then its domain is given by
\begin{equation}
\label{AbstractDef}
 \cD(H) = \{ \psi\in \cD_0 ; \ \partial_n\psi[\gamma\phi] +
    q_2 [ \psi,\phi ] = 0, \ \forall \phi \in \cD_{Q^{(2)}} \}  \ .
\end{equation}
\end{prop}
Due to the presence of the abstract Green's operator this characterisation
of the domain is not very explicit. Our aim is to show that, in certain cases,
the domain \eqref{AbstractDef} coincides with \eqref{bcinterval}. The domains
$\cD_{2}(P,L)$, however, are subspaces of $H^2 (D)$, a property that does not
immediately follow from \eqref{AbstractDef}. In the theory of partial
differential equations this question is well known as the problem of elliptic
regularity (see, e.g., \cite{GilTru83}).
\begin{defn}
\label{DEFRegular}
The quadratic form $Q^{(2)}_{P,L}$ is called {\it regular}, iff its associated
self-adjoint operator $H$ has a domain $\cD(H)\subset H^{2}(D)$.
\end{defn}
As the form \eqref{Qform2} involves a boundary integral, in addition to
regularity we have to impose a (somewhat mild) condition on the projectors
$P$ ensuring that the kernel of the operator $\Pi$ is under sufficient
control.
\begin{lemma}
\label{dense1}
Let $P:(0,l)\to \M(4,\kz)$ be such that its matrix entries are in
$C^1(0,l)$, then $\ran(\gamma|_{\cD_{Q^{(2)}}})$ is dense in $\ker\Pi$ with respect 
to the norm of $L^2 (0,l)\otimes\kz^4$.
\end{lemma}
\begin{proof}
As $C_0^\infty(0,l)\otimes\kz^4\subset L^2(0,l)\otimes\kz^4$ is dense, whenever 
$\chi\in\ker\Pi\subset L^2(0,l)\otimes\kz^4$ there exists a sequence 
$\{\chi_n\}\subset C_0^\infty(0,l)\otimes\kz^4$ that converges to $\chi$. 
Moreover, any $\chi_n\in C_0^\infty(0,l)\otimes\kz^4$ can be extended to some 
$\psi_n\in H^1(D)$, such that $\chi_n =\psi_{n,bv}$.

Using the orthogonal complement $\Pi^\perp$ to the projector $\Pi$ we note that,
by the assumption in the lemma, $\Pi^\perp\chi_n\in C_0^1 (0,l)\otimes\kz^4$. 
Again, $\Pi^\perp\chi_n$ can be extended to a function $\phi_n\in H^1(D)$, such 
that $\Pi^\perp\chi_n =\phi_{n,bv}$. By construction, $P(y)\phi_{n,bv}(y) =0$ so 
that indeed $\phi_n\in\cD_{Q^{(2)}}$. Therefore, identifying $\phi_{n,bv}$ with
$\gamma\phi_n$ we conclude that $\Pi^\perp\chi_n\in\ran(\gamma|_{\cD_{Q^{(2)}}})$.

Moreover, as $\Pi$ is assumed to be bounded in operator norm there exits 
$K>0$ such that
\begin{equation}
 \| \Pi^\perp\chi_{n} - \chi \|_{L^{2}(0,l)\otimes\kz^4} = 
 \| \Pi^\perp\bigl(\chi_n - \chi\bigr)\|_{L^{2}(0,l)\otimes\kz^4} \leq K \, 
 \|\chi_n -\chi\|_{L^{2}(0,l)\otimes\kz^4} \to 0\ ,
\end{equation}
as $n\to\infty$. Thus, $\ran(\gamma|_{\cD_{Q^{(2)}}})$ is dense in $\ker\Pi$.
\end{proof}
In the regular case, and when the matrix entries of $P$ are of class $C^1$, 
we can specify the domain of the operator $H$ more explicitly.
\begin{theorem}
\label{2Laplace}
Suppose that the matrix entries of $P:(0,l)\to \M(4,\kz)$ are in $C^1(0,l)$
and that the quadratic form $Q^{(2)}_{P,L}$ is regular. Then the unique 
self-adjoint, semi-bounded operator $H$ that is associated with this form is 
the two-particle Laplacian $-\Delta_2$ with domain $\cD_2 (P,L)$.
\end{theorem}
\begin{proof}
Since in the regular case any $\psi\in\cD(H)$ is in $H^2(D)$, Green's operator
$\partial_n$ is the standard normal derivative, see \eqref{normalderiv}. This
would allow to state the `boundary condition' contained in \eqref{AbstractDef}
immediately in an explicit way.

Following the one-particle case developed in \cite{Kuc04}, however, we shall 
now proceed in a more direct way as this will confirm the
operator $H$ too. For this we choose $\psi$ in \eqref{abstractBVP}
to be smooth and compactly supported in $D$, vanishing in neighbourhoods of
$\partial D$ such that $\psi_{bv}(y)=0$ for all $y\in [0,l]$. Thus
\begin{equation}
 \langle \chi,\psi\rangle_{L^2(D)} =\  \int_0^l\int_0^l \bigl(\bar{\phi}_x(x,y)\,
 \psi_x(x,y)\ + \bar{\phi}_y(x,y)\,
 \psi_y(x,y)\bigr)\ \ud x\,\ud y \ .
\end{equation}
An integration by parts then yields
\begin{equation}
\label{partintQ}
 \langle \chi,\psi\rangle_{L^2(D)} = \int_0^l\int_0^l \bigl(-\bar{\phi}_{xx}(x,y)
 -\bar{\phi}_{yy}(x,y)\bigr)\psi(x,y)\ \ud x\,\ud y \ ,
\end{equation}
so that $\chi=H\phi=-\Delta_2\phi$. Hence the operator $H$ acts as a
two-particle Laplacian, and every $\phi\in\cD(H)$ must be in
$\cD(-\Delta^{\ast}_{2,0})$. Now, we choose $\psi\in\cD_{Q^{(2)}}$ that is non-zero
in a neighbourhood of $\partial D$. Then, in addition to the right-hand side
of \eqref{partintQ}, an integration by parts yields the term
\begin{equation}
\label{orthocond}
 -\int_0^l \langle \phi'_{bv}(y)+L(y)\phi_{bv}(y),\psi_{bv}(y) \rangle_{\kz^4}
  \ \ud y = -\langle \phi'_{bv} +L\phi_{bv},\psi_{bv} \rangle_{L^2(0,l)\otimes\kz^4}\ ,
\end{equation}
which must vanish. Since $L(\cdot)$ is self-adjoint, one can rewrite this term
as
\begin{equation}
 \int_{\partial D}\frac{\partial \bar{\phi}}{\partial n}\,\psi\ \ud s
 + q_2[\phi,\psi]\ ,
\end{equation}
hence its vanishing is precisely a more explicit version of the boundary
condition in \eqref{AbstractDef}.

Furthermore, the condition $P(y)\psi_{bv}(y)=0$ fulfilled by 
$\psi\in\cD_{Q^{(2)}}$ for a.e. $y\in[0,l]$ implies that $\psi_{bv}$ is in the
kernel of the orthogonal projector $\Pi$ on $L^2(0,l)\otimes\kz^4$. Hence, 
the vanishing of \eqref{orthocond} for all $\psi\in\cD_{Q^{(2)}}$, together with 
the fact that by Lemma~\ref{dense1} $\ran(\gamma|_{\cD_{Q^{(2)}}})\subset\ker\Pi$ 
is dense, implies that $\phi'_{bv}+L(\cdot)\phi_{bv}$ is in the kernel of 
$\Pi^\perp$, or
\begin{equation}
\label{bcregular}
 Q(y)\phi'_{bv}(y)+Q(y)L(y)\phi_{bv}(y)=0 \ .
\end{equation}
Furthermore, as $L(y)$ is an endomorphism of $\ran Q(y)\subseteq\kz^4$, a 
comparison with \eqref{bcinterval} shows that $\cD(H)=\cD_2(P,L)$.
\end{proof}
\begin{rem}
In the general, not necessarily regular, case the integration by parts leading
to \eqref{partintQ} is not possible and the boundary term \eqref{orthocond}
cannot be expressed in the same form as an integral. Nevertheless, weak
derivatives and the abstract version of Green's operator allow to interpret
the abstract boundary condition $\partial_n\psi[\gamma\phi]+q_2 [\psi,\phi]=0$
as a distributional variant of \eqref{bcregular}. Regardless of regularity,
we shall therefore also use $\cD_2(P,L)$ to denote the domain of the
self-adjoint operator $H$ associated with the quadratic form $Q^{(2)}_{P,L}$.
\end{rem}
So far we did not consider to what extent the representation of the
quadratic form in terms of the maps $P$ and $L$, see \eqref{Qform2} and 
\eqref{Defquad}, is unique.
\begin{prop}
\label{uniquePL}
Suppose that the matrix entries of $P:(0,l)\to \M(4,\kz)$ are in $C^1(0,l)$. 
Then the parametrisation of the quadratic form $Q^{(2)}_{P,L}$ in terms of $P$ 
and $L$ according to \eqref{Qform2} and \eqref{Defquad} is unique with this
property.
\end{prop}
\begin{proof}
The characterisation \eqref{Defquad} of a domain $\cD_{Q^{(2)}}$ involves only 
$P$. Suppose that a given domain can be characterised by two different maps 
$P_j :(0,l)\to \M(4,\kz)$, $j=1,2$, both of which with matrix entries in 
$C^1(0,l)$. The associated projection operators $\Pi_j$ on 
$L^2(0,l)\otimes\kz^4$ are, therefore, different implying 
$\ker\Pi_1\neq\ker\Pi_2$. We can hence assume that there exists 
$\chi\in\ker\Pi_1$ such that $\chi\not\in\ker\Pi_2$. Now, following 
Lemma~\ref{dense1} there exists a sequence $\{\phi_n\}$ in $\cD_{Q^{(2)}}$
such that $\phi_{n,bv}$ converges to $\chi$. Moreover, following our assumption
$\phi_n\in\cD_{Q^{(2)}}$ means that $\phi_{n,bv}\in\ker\Pi_1\cap\ker\Pi_2$. 
However, $\chi\not\in\ker\Pi_2$ contradicts the fact that the 
$\phi_{n,bv}\in\ker\Pi_2$ converge to $\chi$.
 
Now assume that a domain $\cD_{Q^{(2)}}$ (with a unique $C^1$-map $P$) is given, 
but the form \eqref{Qform2} can be characterised by two different maps
$L_j :(0,l)\to \M(4,\kz)$, $j=1,2$, yielding two different (bounded and
self-adjoint) operators $\Lambda_j$ on $L^2(0,l)\otimes\kz^4$. Hence
\begin{equation}
 \langle \phi_{bv},\bigl( \Lambda_1 -\Lambda_2 \bigr)\phi_{bv} 
 \rangle_{L^2(0,l)\otimes\kz^4} = 0\ ,\quad\text{for all}\  \phi\in\cD_{Q^{(2)}}\ .
\end{equation}
Again following Lemma~\ref{dense1}, and using that, by definition,  $L_j(y)$ 
vanishes on $\bigl(\ker P(y)\bigr)^\perp$, this implies $\Lambda_1 =\Lambda_2$.
\end{proof}
\begin{rem}
In the regular case, when the associated operators are two-particle Laplacians
with domains $\cD_2(P,L)$, the same uniqueness results holds for the operators,
as the association between closed, semi-bounded quadratic forms and 
semi-bounded, self-adjoint operators is one-to-one \cite{Kat66}.
\end{rem}

In a couple of standard cases it is well known that the quadratic form
\eqref{Qform2} is regular, including the forms associated with the following 
operators:
\begin{enumerate}
\item A Dirichlet-Laplacian, in which case $P(y)=\eins_4$ for all $y\in [0,l]$.
\item A Neumann-Laplacian, where $P(y)=0=L(y)$ for all $y\in [0,l]$.
\item A mixed Dirichlet-Neumann Laplacian, where $P(y)$ is independent of $y$ 
and diagonal such the diagonal entries are either zero or one. Moreover,
$L(y)=0$ for all $y\in [0,l]$. In such a case Dirichlet boundary conditions 
are imposed on the parts of the boundary that, via \eqref{bvinterval}, 
correspond to a one on the diagonal of $P$, and Neumann boundary conditions on 
the remaining parts.
\item A Laplacian with standard Robin boundary condition follows when $P(y)=0$ 
for all $y\in [0,l]$ and  $L=\alpha\eins_4$, where $\alpha>0$. In that case the 
boundary conditions in \eqref{bcinterval} reduce to 
$\psi_{bv}' (y)+\alpha\psi_{bv}(y)=0$.
\end{enumerate}
We are now going to establish regularity in a further class of examples.
Suppose that $P$ takes the following block-diagonal form,
\begin{equation}
\label{Passum}
 P(y) = \begin{pmatrix} \tilde P(y) & 0 \\ 0 & \tilde P(y)\end{pmatrix} \ .
\end{equation}
This structure will be necessary in the case of identical particles, see
Section~\ref{sec3}. We also assume that the matrix entries of $\tilde P$ are 
in $C^3(0,l)$ and possess extensions of class $C^3$ to some interval
$(-\eta,l+\eta)$, $\eta>0$. Then the rank of $\tilde P(y)$, which is either
zero, one or two, is the same for all $y\in[0,l]$. In the case $\rank\tilde P=0$
the only possible solution for $\tilde P$ is to be zero, and for  
$\rank\tilde P=2$ the projector has to be $\eins_2$. These to cases are
covered by the examples 1.\ and 2.\ above. When $\rank\tilde P=1$, the projector
is of the form
\begin{equation}
\label{rkPone}
 \tilde P(y) = \begin{pmatrix} \beta(y) & \bar\gamma(y) \\ 
               \gamma(y) & 1-\beta(y) \end{pmatrix} \ ,
\end{equation}
where $0\leq\beta(y)\leq 1$ and $|\gamma(y)|^2 =\beta(y)-\beta^2(y)$, i.e., 
when $\gamma(y)=0$, $\beta(y)$ must be either one or zero. Hence, demanding 
that $\gamma(y)\to 0$ as $y\to 0$ and as $y\to l$, $\tilde P(y)$ approaches 
one of the two cases 
\begin{equation}
 \begin{pmatrix} 1 & 0 \\ 0 & 0 \end{pmatrix} \ ,\quad
 \begin{pmatrix} 0 & 0 \\ 0 & 1 \end{pmatrix} \ . 
\end{equation}
We shall indeed suppose that on small intervals $[0,\varepsilon_1]$ and 
$[l-\varepsilon_2,l]$, with $\varepsilon_j>0$, the projector $\tilde P(y)$ 
assumes this form as well as that $L(y)=0$. This means that the conditions 
$P(y)\phi_{bv}(y)=0$ and $Q(y)\phi'_{bv}(y)$ imply mixed Dirichlet-Neumann 
conditions in neighbourhoods of two of the corners of $\partial D$ and either 
Dirichlet or Neumann conditions in neighbourhoods of the remaining two 
corners. In such cases we are able to prove regularity.
\begin{theorem}
\label{TheoremPL(y)1}
Let $L$ be Lipschitz continuous on $[0,l]$ and let $P$ be of the block-diagonal
form \eqref{Passum}. Assume that the matrix entries of $\tilde P$ are in 
$C^3(0,l)$ and possess extensions of class $C^3$ to some interval
$(-\eta,l+\eta)$, $\eta>0$. Moreover, when 
$y\in [0,\varepsilon_{1}]\cup [l-\varepsilon_{2},l]$, with some 
$\varepsilon_{1},\epsilon_{2} > 0$, suppose that $L(y)=0$ and that $\tilde P(y)$
is diagonal with diagonal entries that are either zero or one. Then the 
quadratic form $Q^{(2)}_{P,L}$ is regular.
\end{theorem}
The proof is rather technical and we therefore defer it to the
appendix. 

We shall finally identify the self-adjoint realisations of the two-particle
Laplacian that represent actual interactions. For this purpose we have to 
identify among the self adjoint realisations $(-\Delta_2,\cD_2(P,L))$ those 
that arise as lifts of a one-particle Laplacians to $\cH_{2}$. Hence, let 
$P^{(1)}$ and $L^{(1)}$ be a projector and a self-adjoint operator on $\kz^2$, 
defining a domain \eqref{1partBCalt} for a one-particle Laplacian on $[0,l]$. 
The associated quadratic form and its domain is given by \eqref{Qform1} and
\eqref{Qformdomain}. Its lift to $\cH_{2}$ then is the quadratic form
\eqref{Qform2} defined on the domain \eqref{Defquad}, where
\begin{equation}
\label{PLblock}
 \tilde P(y) = \begin{pmatrix} P^{(1)} & 0 \\ 0 & P^{(1)} \end{pmatrix} 
 \quad\text{and}\quad
 \tilde L(y) = \begin{pmatrix} L^{(1)} & 0 \\ 0 & L^{(1)} \end{pmatrix}
\end{equation}
for all $y\in [0,l]$. Together
with Proposition~\ref{uniquePL} this now leads to a complete characterisation 
of the interacting vs.\ the non-interacting representations of the two-particle 
Laplacian.
\begin{prop}
Suppose that the matrix entries of $P:(0,l)\to \M(4,\kz)$ are in $C^1(0,l)$. 
Then the two-particle Laplacian $-\Delta_{2}$ with domain $\cD_{2}(P,L)$ 
represents no interactions, iff $P$ and $L$ are block-diagonal as in 
\eqref{PLblock} and are independent of $y$.
\end{prop}
We remark that, clearly, a Dirichlet-, a Neumann- or a standard
Robinon-Laplacian on $D$ represent no interactions. 
%
%
\subsection{Two particles on a general compact metric graph}
The construction of self-adjoint realisations of the two-particle
Laplacian on a general compact, metric graph will be based on the above
results for the interval. It will use the same methods and mainly involves
the introduction of a suitable notation.

For convenience we shall again first treat two distinguishable particles,
for which the two-particle Hilbert space is
\begin{equation}
\label{HilbertSpaceTwoParticles}
 \cH_2 = L^2(\Gamma)\otimes L^2(\Gamma) = \bigoplus_{e_1 e_2} L^2 (D_{e_1 e_2})\ .
\end{equation}
Here $D_{e_1 e_2}=(0,l_{e_1})\times (0,l_{e_2})$ denotes the rectangle on which
the component $\psi_{e_1 e_2}\in L^2 (D_{e_1 e_2})$ of
$\Psi =(\psi_{e_1 e_2})\in\cH_2$ is defined. In this sense one can view the 
vectors in $\cH_2$ as functions on the disjoint union
\begin{equation}
\label{DomainGeneralGraph}
 D_\Gamma := \dot{\bigcup_{e_1 e_2}} D_{e_1 e_2}
\end{equation}
of $E^2$ rectangles. With this motivation in mind we shall also denote $\cH_2$
as $L^2 (D_\Gamma)$, and other function spaces such as $H^m(D_\Gamma)$
accordingly. The trace map $\gamma$ defined on $H^1(D_\Gamma)$ then assigns
each function its boundary values, i.e., its values on the disjoint union
\begin{equation}\label{DomainGeneralGraph2}
 \partial D_\Gamma := \dot{\bigcup_{e_1 e_2}} \partial D_{e_1 e_2}
\end{equation}
of the boundaries of the rectangles $D_{e_1 e_2}$. 

A two-particle Laplacian $-\Delta_{2,0}$ can be defined on the domain
$\cD(-\Delta_{2,0}) = C^\infty_0 (D_\Gamma)$, on which it is a symmetric,
non-self adjoint operator. Following \eqref{NLaplace}, it acts on a state
$\Psi$ in its domain as
\begin{equation}
 (-\Delta_{2,0}\Psi)_{e_1 e_2}(x_{e_1},y_{e_2}) =
 -\psi_{e_1 e_2,xx}(x_{e_1},y_{e_2}) - \psi_{e_1 e_2,yy}(x_{e_1},y_{e_2})\ .
\end{equation}
The domain $\cD(-\Delta_{2,0}^\ast)$ of its adjoint $-\Delta_{2,0}^\ast$ is the
immediate analogue of \eqref{adjdom}.

Simplifying notation we shall sometimes use a rescaling of variables
in that we set
\begin{equation}
\label{lscale}
 \psi_{e_1 e_2}(x_{e_1},y_{e_2}) = \psi_{e_1 e_2}(l_{e_1}x,l_{e_2}y)
\end{equation}
with $x,y\in (0,1)$. This then requires to modify the $4E^2$ boundary values
of functions $\Psi\in H^1(D_\Gamma)$ and derivatives of functions
$\Psi\in H^2(D_\Gamma)$ as compared to \eqref{bvinterval},
\begin{equation}
\label{graphbv}
 \Psi_{bv}(y) =
 \begin{pmatrix}\sqrt{l_{e_2}}\psi_{e_1 e_2}(0,l_{e_2}y) \\
 \sqrt{l_{e_2}}\psi_{e_1 e_2}(l_{e_1},l_{e_2}y) \\
 \sqrt{l_{e_1}}\psi_{e_1 e_2}(l_{e_1}y,0) \\
 \sqrt{l_{e_1}}\psi_{e_1 e_2}(l_{e_1}y,l_{e_2})
 \end{pmatrix} \qquad\text{and}\qquad
 \Psi'_{bv}(y) =
 \begin{pmatrix}\sqrt{l_{e_2}}\psi_{e_1 e_2,x}(0,l_{e_2}y) \\
 -\sqrt{l_{e_2}}\psi_{e_1 e_2,x}(l_{e_1},l_{e_2}y)\\
 \sqrt{l_{e_1}}\psi_{e_1 e_2,y}(l_{e_1}y,0) \\
 -\sqrt{l_{e_1}}\psi_{e_1 e_2,y}(l_{e_1}y,l_{e_2})\end{pmatrix} \ .
\end{equation}
Here $y\in [0,1]$ and the indices $e_1 e_2$ run over all $E^2$ possible pairs
with $e_1 ,e_2 =1,\dots,E$. As in the case of two particles on an interval,
$\Psi_{bv}\in L^2(0,1)\otimes\kz^{4E^2}$ is a convenient way to parametrise the 
trace $\gamma\Psi\in L^2(D_\Gamma)$ of $\Psi\in H^1(D_\Gamma)$.

In order to formulate appropriate boundary conditions we introduce maps
$P,L: [0,1] \to \M(4E^2,\kz)$ such that
\begin{enumerate}
\item $P(y)$ is an orthogonal projector,
\item $L(y)$ is a self-adjoint endomorphism on $\ker P(y)$,
\end{enumerate}
for a.e. $y \in [0,1]$; moreover $Q(y)=\eins_{4E^2}-P(y)$. As previously, these
maps are required to be bounded and measurable. Moreover, operators $\Pi$ 
and $\Lambda$ can be defined on $L^2(0,1)\otimes\kz^{4E^2}$ through 
$(\Pi\chi)(y):=P(y)\chi(y)$ and $(\Lambda\chi)(y):=L(y)\chi(y)$, respectively. 
As $P$ and $L$ are bounded and measurable functions on $[0,1]$, these operators 
are bounded. Again, $\Pi$ is a projector and $\Lambda$ is self-adjoint.

This notation now allows to define the following domains for two-particle
Laplacians,
\begin{equation}
\label{bcgraph}
\begin{split}
 \cD_2 (P,L) := \{
       &\Psi\in H^2(D_\Gamma);\ P(y)\Psi_{bv}(y)=0\ \text{and}\\
       &\quad Q(y)\Psi'_{bv}(y)+L(y)Q(y)\Psi_{bv}(y)=0\ \text{for a.e.}\
          y\in [0,1] \}
\end{split}
\end{equation}
in close analogy to \eqref{bcinterval}. A representation of these domains
in terms of maps $A,B:[0,1]\to\M(4E^2,\kz)$ can be provided in the same
manner as for two particles on an interval.

In order to address the question of self-adjointness we follow the strategy
outlined above and first generate a suitable quadratic form,
\begin{equation}
\label{Qform2graph}
\begin{split}
 Q^{(2)}_{P,L}[\Psi] 
   &:= \langle  \nabla\Psi,\nabla\Psi \rangle_{L^2 (D_\Gamma)} - 
         \langle \Psi_{bv},\Lambda\Psi_{bv} \rangle_{L^2(0,1)\otimes\kz^{4E^2}} \\
   &=\sum_{e_1 ,e_2 =1}^E \int_{0}^{l_{e_2}}\int_0^{l_{e_1}}\Bigl( \bigl|
      \psi_{e_1 e_2,x}(x,y) \bigr|^2 + \bigl| \psi_{e_1 e_2,y}
      (x,y) \bigr|^2 \Bigr)\ \ud x\,\ud y \\
   &\qquad\qquad -\int_0^1 \langle\Psi_{bv}(y),L(y)\Psi_{bv}(y)
       \rangle_{\kz^{4E^2}} \ \ud y \ .
\end{split}
\end{equation}
This allows us to generalise Theorems~\ref{2quadform} and
\ref{2Laplace} to general compact metric graphs.
\begin{theorem}
\label{2quadformgraph}
Given maps $P,L:[0,1]\to \M(4E^2,\kz)$ as above that are bounded and measurable,
the quadratic form \eqref{Qform2graph} with domain
\begin{equation}
\label{Defquadgraph}
 \cD_{Q^{(2)}} = \{ \Psi \in H^1(D_\Gamma);\ P(y)\Psi_{bv}(y)=0\ \text{for a.e.}\
 y\in [0,1] \}
\end{equation}
is closed and semi-bounded.
\end{theorem}
The proof can be taken over verbatim from Theorem~\ref{2quadform} when the 
notation is adapted to the slightly more complex situation of a general graph.

The procedure to extract the self-adjoint operator $H$ with domain $\cD(H)$
associated with the quadratic form $Q^{(2)}_{P,L}$ will be the same as in
Section~\ref{2partint}. The notation set out above allows us to copy
the procedure verbatim, only replacing $D$ by $D_\Gamma$. This first leads
to an abstract Green's operator associated with the trace map
$\gamma:H^1(D_\Gamma)\to L^2(\partial D_\Gamma)$, and then to the equivalent
of Proposition~\ref{abstrdomain}.
\begin{prop}
\label{abstrdomainG}
Let $H$ be the unique self-adjoint, semi-bounded operator corresponding to the
quadratic form $Q^{(2)}_{P,L}$. Then its domain is given by
\begin{equation}
\label{AbstractDefG}
 \cD(H) = \{ \Psi\in \cD_0 ; \ \partial_n \Psi[\gamma\Phi] +
    q_2 [ \Psi,\Phi ] = 0, \ \forall \Phi \in \cD_{Q^{(2)}} \}  \ .
\end{equation}
\end{prop}
Again, one would like to have a more explicit description of the domain. This
is available when $\Psi\in H^2(D_\Gamma)$ as then
\begin{equation}
\label{normalderivG}
 \partial_n \Psi(\gamma\Phi) = \sum_{e_1 e_2}\int_{\partial D_{e_1 e_2}}
 \frac{\partial \bar{\psi}_{e_1 e_2}}{\partial n}\,\phi_{e_1 e_2}\ \ud s\ .
\end{equation}
The Definition~\ref{DEFRegular} of regularity can be extended in an obvious 
way, specifying cases in which the domain of the self-adjoint operator can be 
given more explicitly by using \eqref{normalderivG}. As a result, one obtains
the following statement, which can be proved in complete analogy to 
Lemma~\ref{dense1}. 
\begin{lemma}
\label{dense2}
Let $P:(0,1)\to \M(4E^2,\kz)$ be such that its matrix entries are in
$C^1(0,1)$, then $\ran(\gamma|_{\cD_{Q^{(2)}}})$ is dense in $\ker\Pi$ with respect 
to the norm of $L^2 (0,1)\otimes\kz^{4E^2}$.
\end{lemma}
If the condition in the lemma is fulfilled and the quadratic form is
regular we obtain an analogue of Theorem~\ref{2Laplace}.
\begin{theorem}
\label{2Laplacegraph}
Suppose that the matrix entries of $P:(0,1)\to \M(4E^2,\kz)$ are in $C^1(0,1)$
and that the quadratic form $Q^{(2)}_{P,L}$ is regular. Then the unique 
self-adjoint, semi-bounded operator $H$ that is associated with this form is 
the two-particle Laplacian $-\Delta_2$ with domain $\cD_2 (P,L)$.
\end{theorem}
\begin{proof}
The proof is in close analogy to the proof of Theorem~\ref{2quadformgraph},
and leads to an obvious generalisation of \eqref{partintQ}. Performing the
integration by part with some $\Psi\in\cD_{Q^{(2)}}$ that does not vanish in a
neighbourhood of $\partial D_\Gamma$ one obtains the additional term
\begin{equation}
\label{orthocondgraph}
 -\int_0^1 \langle \Phi'_{bv}(y)+L(y)\Phi_{bv}(y),\Psi_{bv}(y) \rangle_{\kz^{4E^2}}
  \ \ud y = -\langle \Phi'_{bv} +L\Phi_{bv},\Psi_{bv}
  \rangle_{L^2(0,1)\otimes\kz^{4E^2}}\ .
\end{equation}
This is the explicit expression for
$\partial_n \Psi[\gamma\Phi]+q_2 [\Psi,\Phi]$ and is required to vanish. Again,
the fact that $\ran(\gamma|_{\cD_{Q^{(2)}}})$ is dense in $\ker\Pi$ implies the 
condition
\begin{equation}
 Q(y)\Psi'_{bv}(y)+L(y)Q(y)\Psi_{bv}(y)=0\ \text{for a.e.}\ y\in [0,1]
\end{equation}
in \eqref{bcgraph}. 
\end{proof}
Uniqueness of the parametrisation of the quadratic form in terms of maps $P$ 
and $L$, see \eqref{Qform2graph} and \eqref{Defquadgraph} follows in the same
way as on an interval, see Proposition~\ref{uniquePL}.
\begin{prop}
\label{uniquePLgraph}
Suppose that the matrix entries of $P:(0,1)\to \M(4E^2,\kz)$ are in $C^1(0,1)$. 
Then the parametrisation of the quadratic form $Q^{(2)}_{P,L}$ in terms of $P$ 
and $L$ according to \eqref{Qform2graph} and \eqref{Defquadgraph} is unique 
with this property.
\end{prop}
We are now going to establish the regularity of the quadratic form 
$Q^{(2)}_{P,L}$ in a class of examples that are similar to 
Theorem~\ref{TheoremPL(y)1}. We assume that the projectors $P(y)$ possess
a block-diagonal form in analogy to \eqref{Passum}, where the two blocks
$\tilde P$ correspond to the $2E^2$-dimensional subspaces spanned by the upper 
and lower half, respectively, of the components in \eqref{graphbv}. We also 
suppose that the matrix entries of $\tilde P$ are in $C^3(0,1)$ with an 
extension of class $C^3$ to some interval $(-\eta,1+\eta)$, $\eta>0$. As in 
the case of an interval we need to impose either Dirichlet, Neumann or mixed 
Dirichlet-Neumann boundary conditions near corners of the rectangles
$D_{e_1e_2}$. This will be achieved when $\tilde P(y)$ is diagonal and the 
diagonal entries are either zero or one.
\begin{theorem}
\label{TheoremPL(y)2}
Let $L$ be Lipschitz continuous on $[0,1]$ and let $P$ be of the block-diagonal
form \eqref{Passum}. Assume that the matrix entries of $\tilde P$ are in 
$C^3(0,1)$ and possess extensions of class $C^3$ to some interval
$(-\eta,1+\eta)$, $\eta>0$. Moreover, when 
$y\in [0,\varepsilon_{1}]\cup [l-\varepsilon_{2},l]$, with some 
$\varepsilon_{1},\varepsilon_{2} > 0$, suppose that $L(y)=0$ and that 
$\tilde P(y)$ is diagonal with diagonal entries that are either zero or one. 
Then the quadratic form $Q^{(2)}_{P,L}$ is regular.
\end{theorem}
The proof of this result is essentially the same as that of 
Theorem~\ref{TheoremPL(y)1}, see the appendix, and mainly involves suitable 
changes in the notation. 
%
%
%
%
\section{Two identical particles on a graph}
\label{sec3}
Implementing exchange symmetry for two identical particles requires the
bosonic as well as the fermionic two-particle Hilbert spaces $\cH_{2,B}$ and 
$\cH_{2,F}$, respectively, as well as a projection of the Hamiltonian
operators and their domains to $\cH_{2,B/F}$.
%
%
\subsection{Bosons and fermions on a general graph}
The Hilbert spaces $\cH_{2,B/F}$ consist of states $\Psi=(\psi_{e_1 e_2})\in\cH_2$
that are (anti-) symmetric under a particle exchange. This means that the $E^2$
components $\psi_{e_1 e_2}\in L^2(D_{e_1 e_2})$ satisfy the relations
\begin{equation}
 \psi_{e_1 e_2}(x_{e_1},y_{e_2})=\pm\psi_{e_2 e_1}(y_{e_2},x_{e_1})\ ,
\end{equation}
where the '$+$' corresponds to the bosonic case, and the '$-$' applies in
the fermionic case.
Notice that these conditions relate pairs of functions when $e_1\neq e_2$ and 
impose conditions on single functions when $e_1 =e_2$. Therefore, only
$\tfrac{1}{2}E(E+1)$ of the $E^2$ components of $\Psi$ in $\cH_{2,B}$ or 
$\cH_{2,F}$ are independent. In order to keep the notation simpler, however, we 
shall keep all components in what follows. We shall then use $L^2(D_\Gamma)_s$ 
to denote the subspace of symmetric functions (with respect to the particle 
exchange), and $L^2(D_\Gamma)_a$ to denote the subspace of anti-symmetric 
functions i.e., $\cH_{2,B}=L^2(D_\Gamma)_s$ and $\cH_{2,F}=L^2(D_\Gamma)_a$.

When the components are in $H^m (D_{e_1 e_2})$, $m=1,2$, they fulfil
\begin{equation}
\label{Symmetrie_gen}
 \psi_{e_1 e_2,x}(x_{e_1},y_{e_2}) = \pm \psi_{e_2 e_1,y}(y_{e_2},x_{e_1}) \quad\text{and}
 \quad \psi_{e_1 e_2,xx}(x_{e_1},y_{e_2}) = \pm\psi_{e_2 e_1,yy}(y_{e_2},x_{e_1}) \ .
\end{equation}
We denote the spaces of these functions as $H^m(D_\Gamma)_s$ and $H^m(D_\Gamma)_s$,
respectively.

In order to realise a particle exchange symmetry of the quadratic forms we
have to ensure a particular structure of the maps $P$ and $L$. They have to
consist of two identical, diagonal $2E^2\times 2E^2$-blocks,
\begin{equation}
\label{2block}
 M(y) = \begin{pmatrix} \tilde M(y) & 0 \\ 0 & \tilde M(y) \end{pmatrix} \ ,
\end{equation}
where $M(y)$ either denotes $P(y)$ or $L(y)$, compare also \eqref{Passum}. 
Following \eqref{graphbv}, these blocks correspond to the subspaces of 
boundary values spanned by the upper half and the lower half, respectively, 
of the components of $\Psi_{bv}$ or $\Psi_{bv}'$.

Making use of the rescaling \eqref{lscale} and the symmetry relations
\eqref{Symmetrie_gen} we introduce $2E^2$-component boundary values
\begin{equation}
\label{graphbvB}
 \tilde\Psi_{bv}(y) =
 \begin{pmatrix}\sqrt{l_{e_2}}\psi_{e_1 e_2}(0,l_{e_2}y) \\
 \sqrt{l_{e_2}}\psi_{e_1 e_2}(l_{e_1},l_{e_2}y) \end{pmatrix}
 \qquad\text{and}\qquad
 \tilde\Psi'_{bv}(y) =
 \begin{pmatrix}\sqrt{l_{e_2}}\psi_{e_1 e_2,x}(0,l_{e_2}y) \\
 -\sqrt{l_{e_2}}\psi_{e_1 e_2,x}(l_{e_1},l_{e_2}y)\end{pmatrix}
\end{equation}
of functions $\Psi\in H^1(D_\Gamma)_{s/a}$ and derivatives of functions
$\Psi\in H^2(D_\Gamma)_{s/a}$. In both symmetry classes these
boundary values are sufficient to serve as representatives for traces 
$\gamma\Psi$ of functions $\Psi\in H^1(D_\Gamma)_{s/a}$.

With these notions the quadratic form \eqref{Qform2graph} reads
\begin{equation}
\label{Qform2graphB}
\begin{split}
 Q^{(2),B/F}_{P,L}[\Psi] 
  &= 2\sum_{e_1 ,e_2 =1}^E \int_{0}^{l_{e_2}}\int_0^{l_{e_1}}\bigl|
      \psi_{e_1 e_2,x}(x,y) \bigr|^2 \ \ud x\,\ud y \\
  &\quad -2\int_0^1 \langle\tilde\Psi_{bv}(y),\tilde L(y)\tilde\Psi_{bv}(y)
       \rangle_{\kz^{2E^2}} \ \ud y \ ,
\end{split}
\end{equation}
and is defined on either of the domains
\begin{equation}
\label{DefGquadB}
 \cD_{Q^{(2),B/F}} = \{ \Psi\in H^1(D_\Gamma)_{s/a};\ \tilde P(y)\tilde\Psi_{bv}(y)=0\
  \text{for a.e.}\ y \in [0,1]\} \ .
\end{equation}
Correspondingly, the domain \eqref{bcgraph} converts into
\begin{equation}
\label{bcG_B}
\begin{split}
 \cD_{2,B/F}(P,L) := \{
     &\Psi\in H^2(D_\Gamma)_{s/a};\ \tilde P(y)\tilde\Psi_{bv}(y)=0\ \text{and}\\
     &\quad\tilde Q(y)\tilde\Psi'_{bv}(y)+\tilde L(y)\tilde Q(y)
        \tilde\Psi_{bv}(y)=0\ \text{for a.e.}\ y\in [0,1] \} \ .
\end{split}
\end{equation}
Theorem~\ref{2quadformgraph} now carries over immediately to the case of either
two bosons or two fermions on the graph. The self-adjoint operator $H_{B/F}$ 
associated with the quadratic form \eqref{Qform2graphB} on the domain 
\eqref{bcG_B} is identified by the immediate analogue of 
Proposition~\ref{abstrdomainG}.

Finally, in the case of a regular quadratic form we can again identify the
operator and its domain in an explicit way, leading to the main result of 
this section.
\begin{theorem}
\label{2LaplacegraphB}
Suppose that the matrix entries of $\tilde P:(0,1)\to \M(2E^2,\kz)$ are in 
$C^1(0,1)$ and that the quadratic form $Q^{(2),B/F}_{P,L}$ is regular. Then the 
unique self-adjoint, semi-bounded operator $H_{B/F}$ that is associated with 
this form is the bosonic or fermionic two-particle Laplacian $-\Delta_{2,B/F}$ 
with domain $\cD_{2,B/F} (P,L)$.
\end{theorem}
\begin{rem}
Examples of regular quadratic forms follow immediately from 
Theorem~\ref{TheoremPL(y)2}.
\end{rem}
We still need to identify those realisations of the (bosonic or fermionic) 
two-particle Laplacian that represent genuine two-particle interactions. In 
order to find these we decompose the space of boundary values as
\begin{equation}
\label{bcdecomp}
 V=\bigoplus_{e_2 =1}^E V_{e_2} \ ,
\end{equation}
where $V_{e_2}\cong\kz^{2E}$ is the subspace of partial boundary values
$\tilde\Psi_{bv,e_2}$ and $\tilde\Psi'_{bv,e_2}$ of the components in 
\eqref{graphbvB} with fixed $e_2$. Loosely speaking, $V_{e_2}$ contains the 
boundary values of states for one particle on edge $e_1$ under the condition 
that the second particle is at some point $l_{e_2}y$ on egde $e_2$. 
\begin{prop}
Suppose that the matrix entries of $\tilde P:(0,1)\to \M(2E^2,\kz)$ are in 
$C^1(0,1)$. Then the two-particle Laplacian $-\Delta_{2,B/F}$ with domain 
$\cD_{2,B/F}(P,L)$ represents no interactions, iff $\tilde P$ and $\tilde L$ are 
independent of $y$ and are block-diagonal with respect to the decomposition 
\eqref{bcdecomp}, where the blocks are identical and represent corresponding 
one-particle maps.
\end{prop}
\begin{proof}
When $\tilde P,\tilde L$ are independent of $y$ and are block-diagonal with 
respect to \eqref{bcdecomp}, with identical blocks, the domain $\cD_2(P,L)$ can
immediately be split into (identical) one-particle domains. Due to the uniqueness
of the representation that follows from Proposition~\ref{uniquePLgraph} every
non-interacting two-particle Laplacian must be of this form.
\end{proof}

Next we wish to characterise local realisations of the two-particle Laplacian.
We recall that for a single particle on a graph locality means that the
respective maps $P$ and $L$ are block-diagonal with respect to the 
decomposition \eqref{1localbc}. For two particles locality has to be considered
on two levels: one-particle and two-particle locality. The latter should mean 
that two particles only interact when they are either on the same edge, or on 
two edges that are connected in a vertex. On a one-particle level locality 
means that each of the two particles interacts with the outside only at actual 
vertices, in the same way as in a one-particle quantum graph.

In more detail, two-particle locality means that boundary conditions should
only be imposed on those components of $\Psi=(\psi_{e_1 e_2})$ where there
is a vertex connecting the edges $e_1$ and $e_2$. This is implemented by first
ordering the components in \eqref{graphbvB} according to the vertices that
comprise either the initial end (where the coordinate vanishes), or the final 
end (where the coordinate takes its maximal value) of edge $e_1$. This yields a 
decomposition of the space of boundary values \eqref{bcdecomp} as
\begin{equation}
V = \bigoplus_{v\in\cV}V_{v} \ ,
\end{equation}
where $V_{v}$ contains all boundary values at edge ends that are connected
in the vertex $v$. Two-particles now means that 
\begin{equation}
 V_{v}=\bigoplus_{v} V_{local,v}\ .
\end{equation}
Here $V_{local,v}$ consists of the components $\psi_{e_{1}e_{2}}(0,l_{e_2}y)$ or
$\psi_{e_{1}e_{2}}(l_{e_1},l_{e_2}y)$, respectively, where the first variable forms 
the vertex $v$ and the second variable lives on an edge $e_{2}$ connected to 
$e_1$, either in the vertex $v$ or in the other edge end of $e_1$. Locality 
then requires that the matrices $P(y)$ and $L(y)$ are block diagonal with 
respect to this decomposition.
%
%
\subsection{An Example}
In this section we want to illustrate the meaning of the boundary conditions 
introduced in the previous paragraphs in a simple example. For this purpose
it is sufficient to consider only local properties at a single vertex. We 
therefore choose the simplest possible example of two identical particles 
(bosons or fermions) on a graph with 
one vertex and two half-lines attached to it. Although this graph is not 
compact and therefore, strictly speaking, is not covered by our results above, 
these can be carried over in an obvious way.

For convenience we characterise the boundary conditions in terms of the 
matrix-valued maps $A,B:[0,\infty)\to\M(4,\kz)$, see the paragraph below
\eqref{bcinterval}. The graph has two edges, hence a bosonic or fermionic
two-particle state 
$\Psi=(\psi_{e_1 e_2})\in\bigl(L^2 (\rz_+^2)\otimes\kz^4\bigr)_{s/a}$ has four 
components. Boundary values are encoded in
\begin{equation}
 \tilde\Psi_{bv}(y) = 
 \begin{pmatrix} \psi_{11}(0,y) \\ \psi_{21}(0,y) \\ \psi_{12}(0,y) \\ 
 \psi_{22}(0,y) \end{pmatrix} \quad\text{and}\quad
 \tilde\Psi'_{bv}(y) = 
 \begin{pmatrix} \psi_{11,x}(0,y) \\ \psi_{21,x}(0,y) \\ \psi_{12,x}(0,y) \\ 
 \psi_{22,x}(0,y) \end{pmatrix} \ ,
\end{equation}
as each edge has only one finite edge end. We then choose the matrices 
\begin{equation}
\label{Example1}
 A(y) = \begin{pmatrix}
         1 & -1 & 0 & 0 \\
         0 & v(0,y) & 0 & 0 \\
         0 & 0 & 1 & -1 \\
         0 & 0 & 0 & v(0,-y) \end{pmatrix} \quad\text{and}\quad
 B(y)= \begin{pmatrix}
         0 & 0 & 0 & 0 \\
         -1 & -1  & 0 & 0 \\
         0 & 0 & 0 & 0 \\
         0 & 0 & -1 & -1 \end{pmatrix}\ ,
\end{equation}
where $v\in C_0^\infty(\rz^2)$ with $v(x,y)=v(y,x)$. 

The boundary conditions $A(y)\tilde\Psi_{bv}(y)+B(y)\tilde\Psi'_{bv}(y)=0$ imply
that $\psi_{11}(0,y)=\psi_{21}(0,y)$ and $\psi_{12}(0,y)=\psi_{22}(0,y)$. These
conditions ensure continuity of the functions across the vertex in the first 
variable. Due to the particle exchange symmetry this carries over to the other 
variable. Moreover,
\begin{equation}
\begin{split}
 -\psi_{11,x}(0,y)-\psi_{21,x}(0,y) &= -v(0,y)\psi_{21}(0,y) \ , \\
 -\psi_{12,x}(0,y)-\psi_{22,x}(0,y) &= -v(0,-y)\psi_{22}(0,y) \ .
\end{split}
\end{equation}

We now arrange the four functions $\psi_{e_1 e_2}$ on $\rz_+^2$ into a single
function $\psi$ on $\rz^2$ by defining, for $x,y>0$,
\begin{equation}
\label{Example3}
\begin{split}
 \psi(x,y):&=\psi_{11}(x,y) \ ,\\
 \psi(-x,-y):&=\psi_{22}(x,y)\ , \\
 \psi(x,-y):&=\psi_{12}(x,y) \ ,\\
 \psi(-x,y):&=\psi_{21}(x,y) \ .
\end{split}
\end{equation}
Converting the boundary conditions \eqref{Example1} into equivalent conditions
for the function $\psi$ then yields a domain for the (formal) Hamiltonian 
\begin{equation}
 \hat{H}=-\frac{\partial^{2}}{\partial x^{2}}-\frac{\partial^{2}}{\partial y^{2}}
                 +v(x,y)[\delta(x)+\delta(y)] \ ,
\end{equation}
for two bosons or fermions on the real line. 

This example hence illustrates that the two-particle Laplacians introduced above
represent singular two-particle interactions that act when (at least) one 
particle hits a vertex.
%
%
%
%
\section{Spectral properties}
\label{sec4}
We shall now show that the operators constructed in the previous section
possess purely discrete spectra and that their eigenvalues are distributed
according to an appropriate Weyl law. These observations are certainly not
surprising as the operators are composed of Laplacian on a finite number of
rectangles. Still, proofs of these statements seem desirable.

If a two-particle Laplacian $-\Delta_{2}$ is a lift of a one-particle Laplacian
$-\Delta_{1}$, i.e.,
$-\Delta_2 = (-\Delta_1)\otimes\eins_{\cH_1}+\eins_{\cH_1}\otimes(-\Delta_1)$
on the domain $\cD_1 (P,L)\otimes\cD_1 (P,L)$, Theorem VIII.33 in
\cite{ReeSim72} implies for the spectrum of the two-particle operator that
\begin{equation}
 \sigma\bigl(-\Delta_{2} \bigr) = \sigma\bigl(-\Delta_{1} \bigr) +
 \sigma\bigl(-\Delta_{1} \bigr)\ ;
\end{equation}
in particular, $\sigma\bigl(-\Delta_{2} \bigr)$ is purely discrete. (Notice
that here $-\Delta_2$ is defined on the Hilbert space $\cH_2$ for two 
distinguishable particles. A projection to the bosonic or fermionic subspace 
is yet to follow.) The eigenvalues of $-\Delta_{2}$ are therefore of the form
$\lambda_{n,m}=k_n^2 +k_m^2$, when
$\sigma\bigl(-\Delta_{1} \bigr)=\{k_n^2; n\in\nz_0\}$. One can, therefore,
immediately determine the eigenvalue asymptotics for such a two-particle
Laplacian.
\begin{lemma}
Let $-\Delta_2 = (-\Delta_1)\otimes\eins_{\cH_1}+\eins_{\cH_1}\otimes(-\Delta_1)$
be a lift of a one-particle Laplacian $(-\Delta_1,\cD(P,L))$ to the
two-particle Hilbert space $\cH_2$. Then the eigenvalues
$\{\lambda_{n,m}; n,m\in\nz_0\}$ of $-\Delta_2$ are distributed
according to the Weyl law,
\begin{equation}
\label{Weyl2free}
 N_2(\lambda) := \{(n,m)\in\nz^2_0;\ \lambda_{n,m}\leq\lambda\} \sim
 \frac{\mathcal{L}^2}{4\pi}\,\lambda\ ,\qquad \lambda\to\infty\ ,
\end{equation}
where $\mathcal{L}=l_1+\dots+l_E$ is the sum of the edge lengths of the graph.
It is understood that $N_2(\lambda)$ counts eigenvalues with their respective
multiplicities.
\end{lemma}
\begin{proof}
It is known \cite{BolEnd09} that for any self-adjoint realisation of the
one-particle Laplacian on a compact, metric graph with eigenvalues $k_n^2$
the eigenvalue count follows a Weyl law,
\begin{equation}
\label{Weyl1}
 N_1(k) := \{n\in\nz_0;\ k_n^2\leq k^2\} \sim \frac{\mathcal{L}}{\pi}\,k
 \ ,\qquad k\to\infty\ .
\end{equation}
Via a Tauberian theorem \cite{Kar31}, this asymptotic law is equivalent to 
\begin{equation}
\label{heat1}
 \sum_n \ue^{-k_n^2 t} \sim \frac{\mathcal{L}}{\sqrt{4\pi t}}\ ,\qquad
 t\to 0+\ .
\end{equation}
Squaring both sides of \eqref{heat1} and using the equivalence between 
eigenvalue asymptotics and heat-trace asymptotics in the opposite direction 
immediately yields \eqref{Weyl2free}.
\end{proof}
Implementing particle exchange symmetry, one observes that the eigenfunctions
$\psi_n\otimes\psi_m\in\cD_1(P,L)\otimes\cD_1(P,L)$ and $\psi_m\otimes\psi_n$ 
of $-\Delta_2$ are transformed into their (anti-) symmetric versions 
$\psi_n\otimes\psi_m\pm \psi_m\otimes\psi_n$. The multiplicity of the eigenvalue
$\lambda_{n,m}=\lambda_{m,n}$ is, therefore, reduced by a factor of two. Hence,
in the bosonic or fermionic case the Weyl asymptotics read
\begin{equation}
\label{Weyl2freebos}
 N_{2,B/F}(\lambda) \sim \frac{\mathcal{L}^2}{8\pi}\,\lambda\ ,\qquad
 \lambda\to\infty\ .
\end{equation}
We shall now generalise the above statements to all realisations (irrespective
of regularity) of the two-particle Laplacian introduced in the previous section.
\begin{theorem}
\label{vonNeumannBracketing}
A self-adjoint, bosonic or fermionic realisation of the two-particle Laplacian
$-\Delta_{2}$ on a domain $\cD_{2,B/F}(P,L)$ has compact resolvent
and, therefore, possesses a purely discrete spectrum. Moreover, the
eigenvalue asymptotic follow the Weyl law \eqref{Weyl2freebos}.
\end{theorem}
\begin{proof}
The proof is based on a comparison with two simple operators (quadratic forms)
in the spirit of the well-known Dirichlet-Neumann bracketing. Both
comparison operators are lifts of one-particle operators.

The first operator, $(-\Delta_2,\cD_{2,B/F}(P_D,L_D))$, is a lift of the 
one-particle Dirichlet-Lapla\-cian, and is characterised by the projector 
$P_D=\eins_V$ as well as $L_D=0$ on the space \eqref{bcdecomp} of boundary 
values. The second operator, $(-\Delta_2,\cD_{2,B/F}(P_R,L_R))$, is a lift of a 
one-particle Robin-Laplacian, and is characterised by the projector $P_R=0$ and
$L_R=\lambda\eins_V$, where $\lambda$ is the operator norm of the bounded
map $\Lambda$ defined by $L$ (see also \cite{BolEnd09}).

When $(-\Delta_{2},\cD_{2,B/F}(P,L))$ is any of the self-adjoint realisations
introduced previously, the associated quadratic forms obviously satisfy the 
following inclusions of their domains,
\begin{equation}
 \cD_{Q^{(2),B/F}_{P_D,L_D}} \subseteq \cD_{Q^{(2),B/F}_{P,L}} \subseteq \cD_{Q^{(2),B/F}_{P_{R},L_{R}}}
 \ .
\end{equation}
This means, in the sense of \cite{ReeSim78}, that
\begin{equation}
\label{D-Rbracket}
 (-\Delta_2,\cD_{2,B/F}(P_D,L_D)) \geq (-\Delta_{2},\cD_{2,B/F}(P,L)) \geq
 (-\Delta_2,\cD_{2,B/F}(P_R,L_R))\ .
\end{equation}
This `bracketing' is in close analogy to the Dirichlet-Neumann bracketing
(see, e.g., \cite{ReeSim78}): Let $H$ be a self-adjoint, semi-bounded
operator on a Hilbert space $\cH$, and define
\begin{equation}
 \mu_n(H) := \sup_{\varphi_1,\dots,\varphi_{n-1}\in\cH}
 \inf_{\begin{subarray}{1} \psi\in [\varphi_1,\dots,\varphi_{n-1}]\perp \\ 
 \psi\in Q_H,\|\psi\|=1\end{subarray}}\langle\psi,H\psi\rangle_\cH\ .
\end{equation}
Then, \eqref{D-Rbracket} implies that
\begin{equation}
 \mu_n(-\Delta_2)_D \geq \mu_n(-\Delta_2) \geq \mu_n(-\Delta_2)_R \ .
\end{equation}
Using that both the Dirichlet- and the Robin-Laplacian have compact resolvent
one concludes (with Theorem XIII.64 in \cite{ReeSim78}) that
$\mu_n(-\Delta_2)_{R,D}\to\infty$ as $n\to\infty$; hence the same is true for
$\mu_n(-\Delta_2)$. By the same theorem this implies that
$(-\Delta_{2},\cD_{2,B/F}(P,L))$ has compact resolvent.

Furthermore, \eqref{D-Rbracket} implies for the eigenvalue counting functions
that
\begin{equation}
 N_{2,B/F}^D(\lambda)\leq N_{2,B/F}(\lambda)\leq N_{2,B/F}^R(\lambda)\ .
\end{equation}
As both $N_{2,B/F}^D$ and $N_{2,B/F}^R$ count eigenvalues of a two-particle 
Laplacian that is a lift of a one-particle Laplacian they both satisfy the Weyl
asymptotics \eqref{Weyl2freebos}. Hence the same asymptotics hold for
$N_{2,B/F}$.
\end{proof}

\vspace*{0.5cm}

\subsection*{Acknowledgement}
J K would like to thank the {\it Evangelisches Studienwerk Villigst} for 
financial support through a Promotionsstipendium.

\vspace*{0.5cm}

\begin{appendix}
\section{A regularity theorem}
In this appendix we prove Theorem~\ref{TheoremPL(y)1} in Section~\ref{sec3}. 
We recall that given a quadratic form as in Theorem~\ref{2quadform} or 
\ref{2quadformgraph}, our goal is to establish regularity of the associated 
self-adjoint operator $H$ with domain \eqref{AbstractDef} or 
\eqref{AbstractDefG}, respectively. To achieve this we need to show that any 
function $\phi\in\cD(H)$ has $H^2$-regularity. As the domains $D_{e_1 e_2}$ on 
which the functions are defined are rectangles, their boundaries have only 
Lipschitz-regularity and are not smooth. Our approach to this problem utilises 
an effective cut-off of the corners in combination with the standard 
{\it difference quotient technique} to establish regularity (see, e.g., 
\cite{GilTru83}). Difference quotients are defined in an obvious way: With 
$h>0$ and $\phi \in L^{2}(D)$, the difference quotients in the positive and 
negative direction of the $i$-th coordinate, respectively, are 
\begin{equation}
\begin{split}
 D^{+h}_{i}\phi(x) &:= \frac{1}{h}\bigl(\phi(x+he_{i})-\phi(x) \bigr),\\
 D^{-h}_{i}\phi(x) &:= \frac{1}{h}\bigl(\phi(x)-\phi(x-he_{i}) \bigr),
\end{split}
\end{equation}
where $e_{i}$ is the corresponding unit vector. 

When $\phi\in C^1(D)$, limits of difference quotients as $h\to 0$ clearly
yield the corresponding directional derivatives. When a function has weaker
regularity estimates of difference quotients allow to conclude weak 
differentiability, see \cite{Dob05}.
\begin{lemma}
\label{Dobrowolski}
Let $\Omega$ a bounded domain. If $\phi \in L^{2}(\Omega)$ and 
$\|D^{+h}_{i}\phi\|_{L^{2}(\Omega_{0})}\leq K$, uniformly for all compact domains 
$\Omega_{0} \Subset \Omega$ and for all $0 < h \leq h_{0}(\Omega_{0})$, then 
$\phi$ is weakly differentiable with respect to $x_{i}$ and 
$\|\phi_{x_i}\|_{L^{2}(\Omega)} \leq K$.
\end{lemma}
Another useful result is the following \cite{Dob05}.
\begin{lemma}
\label{Dobrowolski2}
Let $\Omega$ be a bounded domain. Then
\begin{equation}
 \|D^{\pm h}_{i}\phi\|_{L^{2}(\Omega^\pm_{0})} \leq \|\phi_{x_i}\|_{L^{2}(\Omega)}, 
 \quad \forall \phi \in H^{1}(\Omega),
\end{equation}
where $\Omega^\pm_{0}\subset\Omega$ is the maximal domain on which 
$D^{\pm h}_{i}(\cdot)$ can be defined.
\end{lemma}
Some obvious properties of difference quotients include a `product rule',
\begin{equation}
\label{ProdRule}
 D^{\pm h}_i(\phi\,\psi)(x) = (D^{\pm h}_i\phi)(x)\,\psi(x) + 
 \phi(x\pm he_i)\,D^{\pm h}_i\psi(x) \ ,
\end{equation}
and `integration by parts'. In one dimension this takes the form
\begin{equation}
\label{IntPart}
 \int_a^b \bigl(D^{+h}_x\phi(x)\bigr)\,\psi(x)\ \ud x = 
 -\int_a^b \phi(x)\,D^{-h}_x\psi(x)\ \ud x\ ,
\end{equation}
when the support of either $\psi$ or $\phi$ is contained in $[a-h,b-h]$.

For ease of notation we restrict the following discussion to the case of
two particles on an interval; the extension to general graphs will be obvious. 
In the case of an interval the two-particle configuration space is the 
rectangle $D$. In order to effectively cut its corners off we choose suitable 
test functions $\tau\in C^\infty(D)$ and show that there exists a constant 
$K>0$ such that for any $\phi\in \cD(H)$ the estimate
\begin{equation}
\label{FinalEstimate}
  \|\tau D^{h}_{i}\nabla\phi\|_{L^{2}(D)} \leq K
\end{equation}
holds for all $h \leq h_{0}$. This then allows to apply 
Lemma~\ref{Dobrowolski} to eventually conclude that $\phi\in H^2(D)$.

For convenience we state here the first result (Theorem~\ref{TheoremPL(y)1})
that we wish to prove in this appendix.
\begin{theorem}
\label{ThmPL(y)1}
Let $L$ be Lipschitz continuous on $[0,l]$ and let $P$ be of the block-diagonal
form \eqref{Passum}. Assume that the matrix entries of $\tilde P$ are in 
$C^3(0,l)$ and possess extensions of class $C^3$ to some interval
$(-\eta,l+\eta)$, $\eta>0$. Moreover, when 
$y\in [0,\varepsilon_{1}]\cup [l-\varepsilon_{2},l]$, with some 
$\varepsilon_{1},\epsilon_{2} > 0$, suppose that $L(y)=0$ and that $\tilde P(y)$
is diagonal with diagonal entries that are either zero or one. Then the 
quadratic form $Q^{(2)}_{P,L}$ is regular.
\end{theorem}
\begin{proof}
We first show regularity on any subdomain of the form 
$D^{'}=[0,l] \times [\epsilon,l-\epsilon]$ with $\epsilon > 0$, leaving 
the discussion of regularity in the corners of the domain $D$ until the end.
Our first tool is the double difference quotient 
\begin{equation}
\label{diffquo}
\begin{split}
 D^{-h}_{y}\tau^{2}D^{+h}_{y}\phi(x,y)= 
  &\frac{1}{h^{2}}\bigl(\tau^{2}(y)\phi(x,y+h)-\tau^{2}(y)\phi(x,y)\\ 
  &-\tau^{2}(y-h)\phi(x,y)+\tau^{2}(y-h)\phi(x,y-h)\bigr)\ ,
\end{split}
\end{equation}
where $\phi\in\cD(H)\subset\cD_{Q^{(2)}}$ and $\tau\in C^{\infty}_{0}(\rz)$ is a 
test function with support in $(0,l)$ such that $\tau|_{[\epsilon,l-\epsilon]}=1$
and $\tau\leq 1$ elsewhere. 
Even though $\phi$ satisfies the boundary condition $P(y)\phi_{bv}(y)=0$,
\eqref{diffquo} does, in general not. This is due to the dependence of the 
matrix $P$ on $y$. Therefore, we introduce a correction function 
$\kappa\in H^1(D)$ such that
\begin{equation}
\label{Pcoordinate1}
 D^{-h}_{y}\tau^{2}D^{+h}_{y}\phi + \kappa \in \cD_{Q^{(2)}} \ .
\end{equation}
We now determine and estimate $\kappa$ and, to this end, insert 
\eqref{Pcoordinate1} for $\psi$ into \eqref{abstractBVP}, 
\begin{equation}
\label{FormInsert}
\begin{split}
 \langle\nabla\phi,
  &\nabla(D^{-h}_{y}\tau^{2}D^{+h}_{y}\phi)\rangle_{L^2(D)}+ 
    \langle\nabla\phi,\nabla\kappa\rangle_{L^2(D)} - \langle \phi_{bv},\Lambda
    \bigl(D^{-h}_{y}\tau^{2}D^{+h}_{y}\phi\bigr)_{bv}\rangle_{L^2(0,l)\otimes\kz^4}\\
  &-\langle\phi_{bv},\Lambda\kappa_{bv}\rangle_{L^2(0,l)\otimes\kz^4} 
     =\langle \chi,D^{-h}_{y}\tau^{2}D^{+h}_{y}\phi\rangle_{L^2(D)}+\langle \chi,
     \kappa\rangle_{L^2(D)} \ .
\end{split}
\end{equation}
Employing an integration by parts \eqref{IntPart}, while taking into account 
that $\tau$ is compactly supported in $y$ and choosing $h$ to be sufficiently 
small, the first term of \eqref{FormInsert} can be re-written as
\begin{equation}
\begin{split}
 \int_{0}^{l}\int_{0}^{l}\nabla\bar{\phi}\nabla \bigl(D^{-h}_{y}\tau^{2}D^{+h}_{y}
 \phi\bigr) \ \ud x\,\ud y =
 &-\int_{0}^{l}\int_{0}^{l}\tau^{2}|D^{+h}_{y}\nabla\phi|^{2}\ \ud x\,\ud y \\
 &-\int_{0}^{l}\int_{0}^{l}(\nabla D^{+h}_{y}\bar{\phi})(\partial_y\tau^{2})
  (D^{+h}_{y}\phi)\ \ud x\,\ud y \ .
\end{split}
\end{equation}
Hence, \eqref{FormInsert} yields
\begin{equation}
\begin{split}
 \| \tau D^{+h}_{y}\nabla\phi \|_{L^2(D)}^2 =
  &-\langle \nabla D^{+h}_{y}\phi,\partial_y(\tau^{2})D^{+h}_{y}\phi\rangle_{L^2(D)}
    +\langle\nabla\phi,\nabla\kappa\rangle_{L^2(D)} \\
  &- \langle \phi_{bv},\Lambda\bigl(D^{-h}_{y}\tau^{2}D^{+h}_{y}\phi\bigr)_{bv}
     \rangle_{L^2(0,l)\otimes\kz^4} - \langle\phi_{bv},\Lambda\kappa_{bv}
     \rangle_{L^2(0,l)\otimes\kz^4} \\
  &- \langle \chi,D^{-h}_{y}\tau^{2}D^{+h}_{y}\phi\rangle_{L^2(D)} - \langle \chi,
     \kappa\rangle_{L^2(D)} \ ,
\end{split}
\end{equation}
which allows the estimate
\begin{equation}
\label{EstimatesForm2}
\begin{split}
 \| \tau D^{+h}_{y}\nabla\phi \|_{L^2(D)}^2 \leq
  &\ \|\tau D^{+h}_{y}\nabla\phi\|_{L^{2}(D)} \,  \|2(\partial_y\tau)D^{+h}_{y}
     \phi\|_{L^{2}(D)}  +\|\nabla\phi\|_{L^2(D)}\,\|\nabla\kappa\|_{L^2(D)} \\
  &+ \bigl|\langle \phi_{bv},\Lambda\bigl(D^{-h}_{y}\tau^{2}D^{+h}_{y}\phi\bigr)_{bv}
     \rangle_{L^2(0,l)\otimes\kz^4}\bigr| + \bigl|\langle\phi_{bv},\Lambda\kappa_{bv}
     \rangle_{L^2(0,l)\otimes\kz^4}\bigr| \\
  &+ \|\chi\|_{L^2(D)}\,\|D^{-h}_{y}\tau^{2}D^{+h}_{y}\phi\|_{L^2(D)}+\|\chi\|_{L^2(D)}
     \,\|\kappa\|_{L^2(D)} \ .
\end{split}
\end{equation}
We now use \eqref{trThm}, following from the trace theorem, to conclude that
\begin{equation}
 \bigl|\langle\phi_{bv},\Lambda\kappa_{bv}\rangle_{L^2(0,l)\otimes\kz^4}\bigr|
 \leq C\,\|\phi\|_{H^{1}(D)} \, \|\kappa\|_{H^{1}(D)}\ ,
\end{equation}
where the constant $C>0$ incorporates the constant $c$ from \eqref{trThm}
as well as the norm of the bounded map $\Lambda$. Furthermore, using the 
Cauchy-inequality
\begin{equation}
\label{Cauchieps}
 |ab| < \epsilon a^{2} + \frac{b^{2}}{4\epsilon}\ , \qquad\forall a,b\in\rz\ ,\ 
 \epsilon > 0 \ ,
\end{equation}
in the first, second and fourth term on the right-hand side of 
\eqref{EstimatesForm2} we arrive at
\begin{equation}
\label{EstimatesFormNeu1}
\begin{split}
 \| \tau D^{+h}_{y}\nabla\phi \|_{L^2(D)}^2
  &\leq c_{1}(\epsilon_{1}) + \epsilon_{1}\,\|\tau\nabla D^{+h}_{y}\phi
    \|^{2}_{L^{2}(D)} + c_{2}(\epsilon_{2}) + \epsilon_{2}\,\|\nabla\kappa
    \|^{2}_{L^{2}(D)} \\
  &\quad + \bigl|\langle \phi_{bv},\Lambda\bigl(D^{-h}_{y}\tau^{2}D^{+h}_{y}\phi
    \bigr)_{bv}\rangle_{L^2(0,l)\otimes\kz^4}\bigr| \\
  &\quad + c_{3}(\epsilon_{3})+\epsilon_{3} \, \bigl(\|\kappa\|^{2}_{L^{2}(D)}+
     \|\nabla\kappa\|^{2}_{L^{2}(D)}\bigr)\\
  &\quad + \|\chi\|_{L^{2}(D)} \, \|D^{-h}_{y}\tau^{2}D^{+h}_{y}\phi
     \|_{L^{2}(D)} + \|\chi\|_{L^{2}(D)} \, \|\kappa\|_{L^{2}(D)} \ .
\end{split}
\end{equation}
Here we kept all terms containing the still unknown function $\kappa$ or
difference quotients of $\nabla\phi$ explicitly, as these are the quantities 
we want to estimate; all other terms are absorbed in the quantities 
$c_j(\epsilon_j)$. 

In order to estimate the fourth, the seventh and the last term on the 
right-hand side of \eqref{EstimatesFormNeu1} we need to determine a suitable 
function $\kappa$ and, in particular, show that the bounds
\begin{equation}
\label{kappabound}
 \|\kappa\|_{L^{2}(D)} \leq K_{1} \quad\text{and}\quad 
 \|\nabla\kappa\|^{2}_{L^{2}(D)} \leq K_2 +K_3\, \|\tau D^{+h}_{y}\nabla\phi
 \|^{2}_{L^{2}(D)}
\end{equation}
hold, where $K_j > 0$ are some constants not depending on $h$. 

In order to characterise $\kappa$, we infer from \eqref{Pcoordinate1} that its 
boundary values have to be such that 
\begin{equation}
\label{kap_bv}
 P(y)\bigl( (D^{-h}_{y}\tau^{2}D^{+h}_{y}\phi)_{bv}(y)+\kappa_{bv}(y) \bigr) = 0\ .
\end{equation}
Expanding the double difference quotient and using $P(y)\phi_{bv}(y)=0$, we 
obtain a condition of which the upper two components read
\begin{equation}
\label{kappa_bv}
 \tilde P(y)\bigl(\tau^{2}(y)\tilde\phi_{bv}(y+h)+\tau^{2}(y-h)
 \tilde\phi_{bv}(y-h)\bigr)+h^2 \tilde P(y)\tilde\kappa_{bv}(y) = 0\ ;
\end{equation}
here we employed the notation $\tilde\phi_{bv}(y)=(\phi(0,y),\phi(l,y))^T$ 
(i.e., a non-rescaled analogue of \eqref{graphbvB}) as well as the 
block-structure \eqref{2block} of $P$. Since $\tilde P(y)$ is a projector, 
this condition is solved by
\begin{equation}
\label{kappa_bv1}
 \tilde\kappa_{bv}(y) = -\frac{1}{h^2}\tilde P(y)\bigl(\tau^{2}(y)
 \tilde\phi_{bv}(y+h)+\tau^{2}(y-h)\tilde\phi_{bv}(y-h)\bigr) \ .
\end{equation}
This, however, only yields the boundary values of the function we wish to find.
Moreover, the negative power of $h$ would inhibit the envisaged bound 
\eqref{FinalEstimate}.

An extension of $\kappa_{bv}$ into the interior of the rectangle $D$ can be 
achieved by making use of the particular structure \eqref{Passum} required 
for the projectors $P(y)$ as well as the assumed regularity of its matrix
entries. This allows us to find functions $a,b\in C^3(D_0)$, where $D_0$ is
an open domain containing $\bar D$, and define
\begin{equation}
 \tilde\cP(x,y) := \begin{pmatrix} a(x,y) & b(x,y) \\ b(l-x,y) & a(l-x,y)  
 \end{pmatrix}\ ,
\end{equation}
in such a way that $\tilde P(y)=\tilde\cP(0,y)$. We noted in 
Section~\ref{sec2} that when $\rank\tilde P\in\{0,2\}$, the only options for
$\tilde P(y)$ are zero or $\eins_2$. When $\rank\tilde P=0$ we hence
can choose $a(x,y)=0=b(x,y)$ for all $(x,y)\in D_0$, and when $\rank\tilde P=2$
a corresponding choice would be $a(x,y)=1$ and $b(x,y)=0$ for all 
$(x,y)\in D_0$. If $\rank\tilde P=1$ and $\tilde P(y)$ is of the form
\eqref{rkPone} we have to pick a function $a\in C^3(D_0)$ that interpolates
between $\beta(y)$ at $x=0$ and $1-\beta(y)$ at $x=l$, and a function 
$b\in C^3(D_0)$ that interpolates between $\bar\gamma(y)$ at $x=0$ and 
$\gamma(y)$ at $x=l$.

Due to the required regularity of the respective functions the following 
Taylor expansions,
\begin{equation}
 \tilde\cP(x,y\pm h) = \tilde\cP(x,y) \pm h\,\tilde\cP_y(x,y) + 
                        h^2\,\tilde\cP^\pm_{R_2}(x,y;h) \ ,
\end{equation}
and
\begin{equation}
 \tau^2 (y-h) = \tau^2 (y) - h\,\tau^2_{R_1}(y;h)\ ,
\end{equation}
hold with remainder terms that are of class $C^1$ in the variable $y$ and
are bounded in $h$. Using these expansions in \eqref{kappa_bv1} yields
\begin{equation}
\label{kappa_bv2}
\begin{split}
 \tilde\kappa_{bv}(y) 
  &= \tau^2 (y)\,\tilde P_y (y)\Bigl( \frac{\tilde\phi_{bv}(y+h)-
     \tilde\phi_{bv}(y-h)}{h}\Bigr) + \tau^2_{R_1}(y;h) \tilde P_y (y)
     \tilde\phi_{bv}(y-h) \\
  &\quad + \tau^{2}(y)\tilde P^+_{R_2}(y;h)\tilde\phi_{bv}(y+h)+\tau^{2}(y-h)
     \tilde P^-_{R_2}(y;h)\tilde\phi_{bv}(y-h)\ .
\end{split}
\end{equation}
We then define the function
\begin{equation}
\label{kappadef}
\begin{split}
 \kappa(x,y) 
  &:= \tau^2 (y)\Bigl( a_y(x,y)\,\frac{\phi(x,y+h)-\phi(x,y-h)}{h} \\
  &\qquad\qquad\ + b_y(x,y)\,
      \frac{\phi(l-x,y+h)-\phi(l-x,y-h)}{h}\Bigr) \\
  &\quad + \tau^2_{R_1}(y;h)\bigr( a_y(x,y)\phi(x,y-h) + b_y(x,y)\phi(l-x,y-h)
      \bigl)\\
  &\quad + \tau^{2}(y)\bigl( a^+_{R_2}(x,y;h)\phi(x,y+h)+b^+_{R_2}(x,y;h)\phi(l-x,y+h)
      \bigl) \\
  &\quad + \tau^{2}(y-h) \bigl( a^-_{R_2}(x,y;h)\phi(x,y-h)+b^-_{R_2}(x,y;h)
      \phi(l-x,y-h)\bigl)\ ,
\end{split}
\end{equation}
whose boundary values indeed satisfy \eqref{kap_bv}.
The regularity of the functions involved implies that $\kappa\in H^1(D)$ and,
thus, $\|\kappa\|_{L^2(D)}$ and $\|\kappa\|_{H^1(D)}$ are finite. Moreover,
since $\phi\in H^1(D)$ and
\begin{equation}
\label{phidiff}
 \frac{\phi(x,y+h)-\phi(x,y-h)}{h} = D^{+h}_y\phi (x,y) +  D^{-h}_y\phi (x,y)\ ,
\end{equation}
Theorem~\ref{Dobrowolski2} implies that $\|\kappa\|_{L^2(D)}$ has an 
$h$-independent upper bound. In the same way the second bound in 
\eqref{kappabound} follows from \eqref{kappadef} and \eqref{phidiff}.

Next we estimate the fifth term on the right-hand side of 
\eqref{EstimatesFormNeu1}. We use the self-adjointness of $\Lambda$ and
perform an integration by parts \eqref{IntPart} as well as employing the 
product rule \eqref{ProdRule} to obtain
\begin{equation}
\begin{split}
 \langle \phi_{bv},\Lambda\bigl(D^{-h}_{y}\tau^{2}D^{+h}_{y}\phi\bigr)_{bv}
 \rangle_{L^2(0,l)\otimes\kz^4} 
 &= -\int_0^l \langle L(y+h)(\tau D^{+h}_{y}\phi_{bv})(y),(\tau D^{+h}_{y}
      \phi_{bv})(y)\rangle_{\kz^4}\ \ud y \\
 &\quad -\int_0^l \langle (\tau (D^{+h}_{y}L)\phi_{bv})(y),(\tau D^{+h}_{y}
     \phi_{bv})(y)\rangle_{\kz^4}\ \ud y \ .
\end{split}
\end{equation}
Noting that $L$ is supposed to be bounded and Lipschitz continuous, the
right hand side can be estimated from above in absolute value by
\begin{equation}
 d_1\,\|\tau D^{+h}_{y}\phi_{bv}\|^{2}_{L^{2}(0,l)\otimes\kz^{4}} + d_2\,\|\phi_{bv}
 \|_{L^{2}(0,l)\otimes\kz^{4}}\,\|\tau D^{+h}_{y}\phi_{bv}\|_{L^{2}(0,l)\otimes\kz^{4}}\ ,
\end{equation}
with suitable constants $d_j >0$. Estimating further, we apply \eqref{normeqI}
to the first term and \eqref{trThm} to the second and obtain the bound
\begin{equation}
 4\,d_1\,\left(\frac{2}{\epsilon_4}\|\tau D^{+h}_{y}\phi\|^{2}_{L^{2}(D)} +
 \epsilon_4\|\nabla(\tau D^{+h}_{y}\phi)\|^{2}_{L^{2}(D)}\right) + d_3\,
 \|\phi_{bv}\|_{L^{2}(0,l)\otimes\kz^{4}}\,\|\tau D^{+h}_{y}\phi\|_{H^1(D)}\ ,
\end{equation}
where $\epsilon_4>0$ is sufficiently small. Eventually, again using \eqref{normeqI},
this can be further bounded by
\begin{equation}
 d_4 + d_5 \,\epsilon_5\, \|\tau D^{+h}_{y}\nabla\phi\|^{2}_{L^{2}(D)}\ ,\quad
 \text{where}\ d_j >0\ .
\end{equation}

It remains to estimate the last-but-one term on the right-hand side of
\eqref{EstimatesFormNeu1},
\begin{equation}
 \|D^{-h}_{y}\tau^{2}D^{+h}_{y}\phi\|_{L^{2}(D)} \leq
 \|\partial_y(\tau^{2}D^{+h}_{y}\phi)\|_{L^{2}(D)} \leq d_6  + d_7\,
 \|\tau^{2}D^{+h}_{y}\nabla\phi\|_{L^{2}(D)} \ .
\end{equation}
The last term on the right-hand side can be estimated with the help of 
\eqref{Cauchieps}, and using that $\tau^2\leq 1$,
\begin{equation}
 \|D^{-h}_{y}\tau^{2}D^{+h}_{y}\phi\|_{L^{2}(D)} \leq
 c_6(\epsilon_6) + d_8\,\epsilon_6 \|\tau D^{+h}_{y}\nabla\phi\|_{L^{2}(D)}^2 \ .
\end{equation}

We now collect all bounds for the terms on the right-hand side of 
\eqref{EstimatesFormNeu1} and subtract all contributions of the form
$\epsilon_j \|\tau D^{+h}_{y}\nabla\phi\|_{L^{2}(D)}^2$. By choosing
$\epsilon_{1},\dots,\epsilon_{6}$ sufficiently small we finally obtain the 
bound \eqref{FinalEstimate}. By applying Lemma~\ref{Dobrowolski} to
$\nabla\phi$ on the domain $D^{'}=[0,l] \times [\epsilon,l-\epsilon]$
we conclude that $\phi_{xy}$ and $\phi_{yy}$ are in $H^2(D')$. Since
$\phi_{xx}+\phi_{yy}=\Delta_2\phi$ is known to be in $L^2(D)$ we conclude that
$\phi\in H^2(D')$. The same argument can now be repeated on a domain
$D^{''}=[\epsilon,l-\epsilon] \times [0,l]$ so that, indeed, $\phi$
has $H^2$-regularity away from small neighbourhoods of the corners of
the rectangle $D$.

As the condition imposed on $P$ implies that close to the corners either
Dirichlet or Neumann or mixed Dirichlet-Neumann boundary conditions are 
imposed, regularity of $\phi$ in neighbourhoods of the corners follows from 
standard results (see, e.g., \cite{Nec67,Dau88}).
\end{proof}
\end{appendix}

\vspace*{1cm}

{\small
\bibliographystyle{amsalpha}
\bibliography{literatur}}
\newcommand{\etalchar}[1]{$^{#1}$}
\def\cprime{$'$} \def\polhk#1{\setbox0=\hbox{#1}{\ooalign{\hidewidth
  \lower1.5ex\hbox{`}\hidewidth\crcr\unhbox0}}}
\providecommand{\bysame}{\leavevmode\hbox to3em{\hrulefill}\thinspace}
\providecommand{\MR}{\relax\ifhmode\unskip\space\fi MR }
\providecommand{\MRhref}[2]{%
  \href{http://www.ams.org/mathscinet-getitem?mr=#1}{#2}
}
\providecommand{\href}[2]{#2}

\end{document}